\documentclass[msnc, 11 pts]{informs3-adjust}

\usepackage{amssymb}
\usepackage{amsfonts}
\usepackage{amsmath}
\usepackage{endnotes}
\usepackage{natbib}
\usepackage{graphicx, subfigure, setspace, fullpage, url}
\usepackage{color}
\usepackage[usenames,dvipsnames]{xcolor}
\usepackage{listings}
\usepackage[colorlinks,linkcolor=blue,
filecolor=black,
citecolor = blue,
urlcolor=black]{hyperref}

\newtheorem{theorem}{Theorem}

\newtheorem{corollary}{Corollary}

\newtheorem{lemma}{Lemma}
\newtheorem{proposition}{Proposition}
\newtheorem{remark}{Remark}
\OneAndAHalfSpacedXI
\let\enotesize=\normalsize
\def\notesname{Endnotes}
\def\makeenmark{$^{\theenmark}$}
\def\enoteformat{\rightskip0pt\leftskip0pt\parindent=1.75em
  \leavevmode\llap{\theenmark.\enskip}}
\bibpunct[, ]{(}{)}{,}{a}{}{,}
\def\bibfont{\small}
\def\bibsep{\smallskipamount}

\ECRepeatTheorems
\numberwithin{equation}{section}
\input{tcilatex}

\begin{document}


\RUNAUTHOR{Li and Prokhorov}

\RUNTITLE{Semiparametric Bounds}

\ARTICLEAUTHORS{\AUTHOR{Zhaolin Li}
\AFF{The University of Sydney Business School, Sydney, NSW2006, Australia, \EMAIL{erick.li@sydney.edu.au}}
\AUTHOR{Artem Prokhorov}
\AFF{The University of Sydney Business School \& CEBA \& CIREQ, Sydney, NSW2006, Australia, \EMAIL{artem.prokhorov@sydney.edu.au}}
} 

\TITLE{Tail Probability and Expected Loss Revisited: \\ Theory and Applications of Semiparametric Bounds}

\ABSTRACT{Many management decisions involve accumulated random realizations for which only the first and second moments of their distribution are available. The sharp Chebyshev-type bound for the tail probability and Scarf bound for the expected loss are widely used in this setting. We revisit the tail behavior of such quantities with a focus on independence. 
Conventional primal-dual approaches from optimization are ineffective in this setting. Instead, we use probabilistic inequalities to derive new bounds and offer new insights. For non-identical distributions attaining the tail probability bounds, we show that the extreme values are equidistant regardless of the distributional differences. 
For the bound on the expected loss, we show that the impact of each random variable on the expected sum can be isolated using an extension of the Korkine identity. We illustrate how these new results open up abundant practical applications, including improved pricing of product bundles, more precise option pricing, more
efficient insurance design, and better inventory management. For example, we establish a new solution to the optimal bundling problem, yielding a 17\% uplift in per-bundle profits, and a new solution to the inventory problem, yielding a 5.6\% cost reduction for a model with $20$ retailers. }



\title{}
\author{}

\maketitle

\bigskip

\section{Introduction}

The exploration of the bounds on tail probability and expected loss
pertaining to the sum of random variables has a long and distinguished
history in statistical theory and in managerial applications. In relation to
a \emph{single} random variable with given mean and variance, Chebyshev's
and Markov's inequalities are the most widely known results pertaining to
tail probability \citep[see, e.g.,][]{M1956}, whereas Scarf's inequality is
a well-known bound on linear expected loss \citep[see, e.g.,][]{Scarf2002OR}. These inequalities found numerous practical applications in such areas as
retail pricing \citep{E2010JEMS, 2018selltomoment}, inventory management %
\citep{Scarf2002OR}, option pricing \citep{LO1987, BP2002option}, insurance
planning \citep{JANSEN1986}, and designing quota-bonus or loan contracts %
\citep{LK2021}. 


For sums of \emph{independent} random variables, \citeauthor{chernoff52}'s (\citeyear{chernoff52}) use of the moment generating
function  marked a milestone as it inspired numerous
subsequent results known as Hoeffding, Azuma, and McDiarmid inequalities,
among others \citep[see,
e.g.,][]{Hoeffding1963,Azuma1967,McDiarmid1989}. 
This has significantly facilitated further development and application of bounds on
tail probabilities 
\citep[see, e.g.,][]{freedman:75, pinelis:94, delapena:99,
PIJ2004, marinelli:24}.
However, practically relevant multivariate extensions of Chebyshev's and Scarf's inequalities in the operations research literature that 
fully exploit independence seem relatively scarce and recent. 
For example, \cite{CHP2022} seem to be the first to 
employ the Chebyshev-type one-sided inequality 
for pricing. 

A natural approach to deriving probabilistic inequalities under independence 
is to use the central limit theorem. However, we are interested in non-asymptotic results and, as we show in simulation experiments, the normal approximation turns out to be poor in our settings. Moreover, the fact that we impose no assumptions on higher-order moments (beyond the mean and
variance) means that we cannot make use of the classical Berry-Esseen
inequality or similar results that provide error bounds for normal approximations 
(see, e.g., Section 27 in \cite{billingsley:95} or Chapter 5 in 
\cite{delaPena/etal:09}). 

This paper reconsiders the behavior of sums of independent random variables under the known mean and variance. We start by showing that in the independently and identically distributed (iid) case, the distribution attaining the Chernoff-type bound is a two-point distribution. The primal-dual approaches, standard in operations research literature, are not effective at fully capturing the independence constraints. We discuss the reasons for this and we show how the Chernov-type bound we derive relate to
the existing results, including the bounds obtained by aggregating the Chebyshev-type inequality. 

We establish that, even with heterogeneous means
and variances, the extreme distribution 
displays the property of \emph{equal range}, which is a new result. The property means that the distributions that achieve the bound have the same gap between the maximum and minimum values, and that having different means and variances does not change this. 
We work out important practical implications of
this new result. For example, the equal range property immediately implies
that a mixed bundling strategy does not strictly outperform a pure bundling
strategy in terms of the worst-case analysis.

When analyzing the bound on the linear loss of the sum, 
we lose the product form characterizing tail probabilities, due to the
piecewise nature of the linear loss function. Therefore, a direct
application of Chernoff-type analysis is impossible. Nonetheless, we are
able to achieve important improvements of the aggregation-based bound. This
is done using two other important results which are new and may be of use in other areas of
statistics.

The first result is to show how Korkine's identity 
\citep[see,
e.g.,][pp.~242-243]{MPF1993} can be used in a multivariate environment to isolate the impact of each random variable on the absolute value of the sum.
A single-variable version of Korkine's identity pertains to the covariance
between the random variable $X$ and the indicator variable $\mathbb{I}%
_{\{X>0\}}$. By maximizing the covariance and keeping the mean of the
indicator unchanged, it is possible to derive the extreme distribution. It turns out
that the same insights apply to the multivariate scenario. In this case, the extreme distribution is a two-point distribution for each variable so that the sum follows a binomial distribution, enabling us to compute the bound on the expected linear loss.

The second result is to obtain the solution of a non-standard
optimization problem arising in this setting. The difficulty here is that
the objective function based on the endogenous binomial distribution is
piecewise with respect to the chosen tail probability. To overcome this
hurdle, we first derive the relationship between the tail probability of
each iid distribution and the tail probability of the sum. Subsequently, we
use log-convexity to show that the optimal tail probability of each iid
distribution is one of the two extreme points. This allows us to disregard
the piecewise nature of the objective function. Along the way, we provide
bounds for quantiles, which incorporate a well known result for the median,
and we show that the equal range condition continues to hold for the
distributions attaining the bound for the linear loss 
under
independent but non-identical distributions.



We recognize that with a sufficiently large $N$, both the tail probability and
expected linear loss behave 
according to the normal distribution. However,
there exist a multitude of practical examples where $N$ is restricted. For
instance, the anti-trust laws regulate insurance firms from entering
multiple end markets. Telecommunication companies offer bundles of only several services such as fixed line phone, broadband
internet, and mobile plan. Travel booking websites offer bundles of airfare, hotel, car rental, and theme park tickets. 
Our numerical examples confirm that when $N$ is moderate, the gap between the new bound and normal
estimate remains significant. Moreover, because the new bounds have neat
closed-form expressions, they are suitable for decision makers who prefer
reliable estimates that require minimum computational or data-collection
effort. For instance, many farmers sign hedge-to-arrival or future
contracts on grain each year. Most grain contracts are traded over the counter and it is common to have knowledge of the first two moments of the historical distribution of a crop.
Thus, closed-form benchmarks can aid farmers/brokers in quickly deciding when to lock-in the grain prices \citep[see, e.g.,][]{LL2024}.

In a different but not unrelated setting, \cite{OSSMO2013Siam} proposed an uncertainty quantification framework for developing tail probability bounds based on the bounded difference property, rather than on the knowledge of the first two moments. They prove that the extreme distribution, i.e. the distribution achieving their bound, is a discrete distribution facilitating the subsequent optimization analysis to compute the bound. For instance, their framework applies to McDiarmid's inequality (which also requires bounded differences). However, 
the ambiguity set implied by the bounded difference property seems less general 
than that implied by the knowledge of mean and variance. 
Moreover, 
the two-point distribution we derive provides an endogenously determined range, 
particularly important for cases of non-equal means and variances. 

We focus on applications in pricing, in the context of bundling and options,
and on applications in inventory management. We work out the details of
three special cases where the use of the new bounds results in sizeable
improvement of the expected profit. 
An important observation in these practical examples is that our derivation
of the bounds involves endogenous Bernoulli distributions. In this setting,
the equal range property displayed by the extreme distribution becomes
crucial. It ensures that the lattices in the support of the multivariate
distributions are in fact squares with the same size despite possibly
unequal means and variances. This property significantly simplifies the
derivation for the sum of multiple non-identical binomial random variables
and has implications for bundle pricing. 

The remainder of this paper is organized as follows. Section \ref%
{S:Benchmark} recaps the benchmark results with a single random variable.
Section \ref{S:Tail} presents a bound on the tail probability associated
with the sum of independent random variables. Section \ref{S:Loss} develops
a bound on the expected absolute value for the same setting. Section \ref%
{S:Application} solves practical problems in relation to bundle pricing,
inventory management, and option pricing based on the newly developed
bounds. Section \ref{S:Conclusion} concludes. An online Supplement contains technical proofs used to derive the results of Section \ref{S:Loss}.

\section{Benchmark Results\label{S:Benchmark}}

We define $\xi =X _{1}+X _{2}+\ldots +X _{N}$ as the sum of $N $ i.i.d.~random
variables (where $N\geq 1$) and let $q$ be a finite constant. The random
variable $X _{n}$, $n=1, \ldots, N$, and the constant $q$ have many
different important interpretations as summarized in Table \ref{Tab:Application}. For each application in Table \ref{Tab:Application}, the assumption of independence is arguably the most prominent in the respective literature. For example, \cite{ibragimov/walden:10} and references therein consider i.i.d.~valuations of goods across customers; \cite{wang/cai:23} assume independence of consumer demand and loans; independence of increments in the log-return processes has been traditionally assumed for option pricing \cite[see, e.g.,][]{broadie/detemple:04}; independence across income sources in loan contracts, across customers' losses given default in loan portfolios and across insurance claim amounts is often assumed automatically \cite[see, e.g.,][]{PARIS:05, GARRIDO/etal:16}.

\begin{table}[tbph]
\caption{Application Examples}
\label{Tab:Application}\centering%
\begin{tabular}{ccc}
Application & $X _{n}$ & $q$ \\ \hline
Bundle pricing & Valuation for good $n$ & Bundle price \\ 
Inventory management & Demand of retailer $n$ & Inventory level \\ 
Option pricing & Price change on day $n$ & Strike price \\ 
Loan contract & Income from source $n$ & Loan amount \\ 
Insurance policy & Damage on asset $n$ & Maximum benefit \\ 
&  & 
\end{tabular}%
\end{table}

We assume that the underlying distribution of $X_{n}$ remains unknown but $%
\mathbb{E}\left( X_{n}\right) \equiv \mu $ and $Var\left( X_{n}\right)
\equiv \sigma ^{2}>0$ are known and finite. \cite{LO1987} refers to this
setting as\emph{\ semiparametric}. We are interested in the lower bound on
the tail probability $\Pr \left( \xi >q\right) $ and in bounds on expected
losses $\mathbb{E}\left( \xi -q\right) ^{+}$ and $\mathbb{E}\left( \xi
-q\right) ^{-}$, where $\left( \cdot \right) ^{+}=\max \left( 0,\cdot
\right) $ and $\left( \cdot \right) ^{-}=\min \left( \cdot, 0 \right) $. 
Since $\left( \xi -q\right) ^{+}$ $=\frac{\xi -q}{2}+\frac{1}{2}\left\vert
\xi -q\right\vert $, it is sometimes more convenient to use the upper bound
of $\mathbb{E}\left( \left\vert \xi -q\right\vert \right) $, instead of $%
\mathbb{E}\left( \xi -q\right) ^{+}$, to solve practical problems.

\subsection{Single Variable}

We first recap the known results with $N=1$ (where we can write $\xi =X $)
as benchmarks.

\begin{lemma}
\label{L:N=1}(Single-Variable Bounds) It holds that (a) $\Pr \left( X
-q>0\right) \geq 1-\frac{\sigma ^{2}}{\left( \mu -q\right) ^{2}+\sigma ^{2}}$
if $\mu >q$; and (b) $\mathbb{E}\left\vert X -q\right\vert \leq \sqrt{\left(
\mu -q\right) ^{2}+\sigma ^{2}}$.
\end{lemma}

\begin{proof}
(a) Let $\mathbb{I}_{\{X >0\}}$ be an indicator function satisfying $\mathbb{%
I}=1$ if $X >0$ and $\mathbb{I}=0$ otherwise. When $\mu >q$, it must hold
that $0<\mu -q\leq \mathbb{E}\left( \left( X -q\right) \cdot \mathbb{I}_{\{X
>q\}}\right) $. By the Cauchy inequality, we have $\mathbb{E}\left(
XY\right) \leq \sqrt{\mathbb{E}\left( X^{2}\right) }\sqrt{\mathbb{E}\left(
Y^{2}\right) }$ for any two random variables $X $ and $Y$, and the equality
holds if $X$ and $Y$ are \emph{linearly dependent} or \emph{comonotone}
(i.e., $X=\lambda Y$, almost surely, where $\lambda $ is a constant). Thus
we can write 
\begin{equation*}
0<\mu -q\leq \mathbb{E}\left( \left( X -q\right) \cdot \mathbb{I}_{\{X
>q\}}\right) \leq \sqrt{\mathbb{E}\left( \left( X -q\right) ^{2}\right) 
\mathbb{E}\left( \mathbb{I}_{\{X >q\}}^{2}\right) }=\sqrt{\left( \left( \mu
-q\right) ^{2}+\sigma ^{2}\right) \Pr (X >q)},
\end{equation*}%
from which we find that $\Pr (X >q)\geq \frac{\left( \mu -q\right) ^{2}}{%
\left( \mu -q\right) ^{2}+\sigma ^{2}}=1-\frac{\sigma ^{2}}{\left( \mu
-q\right) ^{2}+\sigma ^{2}}$. Similarly, when $\mu <q$, we find that%
\begin{equation*}
0<q-\mu \leq \mathbb{E}\left( \left( q-X \right) \cdot \mathbb{I}_{\{q>X
\}}\right) \leq \sqrt{\left( \left( \mu -q\right) ^{2}+\sigma ^{2}\right)
\Pr (q>X )},
\end{equation*}%
yielding $\Pr (q>X )\geq \frac{\left( \mu -q\right) ^{2}}{\left( \mu
-q\right) ^{2}+\sigma ^{2}}$.

(b) We observe that $\mathbb{E}\left\vert X -q\right\vert =\mathbb{E}\sqrt{%
\left( X -q\right) ^{2}}\leq \sqrt{\mathbb{E}\left( X -q\right) ^{2}}=\sqrt{%
\left( \mu -q\right) ^{2}+\sigma ^{2}}$, where we apply Jensen's inequality.
\end{proof}

While Lemma \ref{L:N=1}(a) pertains to the one-sided Chebyshev inequality
\citep[also referred to as Cantelli's inequality; see, e.g.,][]{CHP2022}, Lemma \ref{L:N=1}(b) pertains
to Scarf's inequality \citep{Scarf2002OR} since $\mathbb{E}\left( X -q\right) ^{+}=\frac{1}{2}%
\mathbb{E}\left( X -q+\left\vert X -q\right\vert \right) \leq \frac{\mu -q+%
\sqrt{\left( \mu -q\right) ^{2}+\sigma ^{2}}}{2}$.

In Section \ref{SS:ScarfNew}, we develop a tighter bound than part (b) of
Lemma \ref{L:N=1}. Two of the following three remarks relate to the extreme
distribution implied by Lemma \ref{L:N=1}, that is, to the distribution for
which the equality sign holds.

\begin{remark}
\label{rem:comonotone} In part (a), the comonotonicity condition $\left( X
-q\right) =\lambda \mathbb{I}_{\{X >q\}}$, a.s., yields that (i) $X =q$ when 
$\mathbb{I}_{\{X >q\}}=0$ and (ii) $X -q=\lambda $ when $\mathbb{I}_{\{X
>q\}}=1$, where $\lambda =\frac{\left( \mu -q\right) ^{2}+\sigma ^{2}}{\mu -q%
}$. Thus, the extreme distribution attaining the bound of Lemma \ref{L:N=1}%
(a) is a two-point distribution satisfying: $\Pr \left( X =q\right) =\frac{%
\sigma ^{2}}{\left( \mu -q\right) ^{2}+\sigma ^{2}}$ and $\Pr \left( X =\mu +%
\frac{\sigma ^{2}}{\mu -q}\right) =\frac{\left( \mu -q\right) ^{2}}{\left(
\mu -q\right) ^{2}+\sigma ^{2}}$.
\end{remark}

\begin{remark}
\label{rem:equality} For part (b), equality holds when $\left\vert X
-q\right\vert =\lambda =\sqrt{\left( \mu -q\right) ^{2}+\sigma ^{2}}$, a.s.
The extreme distribution attaining the bound of Lemma \ref{L:N=1}(b) is
therefore also a two-point distribution satisfying: $\Pr \left( X =q\pm 
\sqrt{\left( \mu -q\right) ^{2}+\sigma ^{2}}\right) =\frac{1}{2}\mp \frac{%
\mu -q}{2\sqrt{\left( \mu -q\right) ^{2}+\sigma ^{2}}}$.
\end{remark}


\subsection{Simple Aggregate Results\label{SS:aggregate}}

With $N\geq 2$, a technical shortcut is to regard $\xi $ as a single random
variable with mean $\mathbb{E}\left( \xi \right) =N\mu $ and variance $%
Var\left( \xi \right) =N\sigma ^{2}$. Consequently, we can apply Lemma \ref%
{L:N=1} to obtain the following bounds.

\begin{lemma}
\label{L:N-NonSharp}(Aggregate Bounds) It holds that (a) $\Pr \left( \xi
>q\right) \geq 1-\frac{\sigma ^{2}}{N\left( \mu -\frac{q}{N}\right)
^{2}+\sigma ^{2}}$ if $N\mu >q$, and (b) $\mathbb{E}\left\vert \xi
-q\right\vert \leq \sqrt{N^{2}\left( \mu -\frac{q}{N}\right) ^{2}+N\sigma
^{2}}$.
\end{lemma}

A few observations are noteworthy. First, the term $\frac{\sigma ^{2}}{N(\mu
-\frac{q}{N})^{2}+\sigma ^{2}}$ converges to zero at the speed of $\frac{1}{N%
}$. Second, the extreme distributions attaining the bounds in Lemma \ref%
{L:N-NonSharp} violate the independence constraints even though $Var\left(
\xi \right) =N\sigma ^{2}$ is consistent with independence. Specifically, to
make the proposed bounds sharp, the joint distribution underlying $\xi $
must have only two possible realized values. We can verify that the joint
distribution in the following Table \ref{Tab:JointDistribution} attains the
bound for $\mathbb{E}\left\vert \xi -q\right\vert $ shown in Lemma \ref%
{L:N-NonSharp}.
\begin{table}[tbph]
\caption{A Joint Distribution for Lemma \protect\ref{L:N-NonSharp}(b).}%
\centering%
\begin{tabular}{cc|cccc}
Case Index & Probability & $X_{1}$ & $X_{2}$ & $...$ & $X_{N}$ \\ \hline
$1$ & $\frac{\gamma }{N}$ & $\xi _{L}$ & $0$ & $...$ & $0$ \\
$2$ & $\frac{\gamma }{N}$ & $0$ & $\xi _{L}$ & $...$ & $0$ \\
$...$ & $...$ & $0$ & $0$ & $...$ & $0$ \\
$N$ & $\frac{\gamma }{N}$ & $0$ & $0$ & $...$ & $\xi _{L}$ \\
$N+1$ & $\frac{1-\gamma }{N}$ & $\xi _{H}$ & $0$ & $...$ & $0$ \\
$N+2$ & $\frac{1-\gamma }{N}$ & $0$ & $\xi _{H}$ & $...$ & $0$ \\
$...$ & $...$ & $0$ & $...$ & $...$ & $0$ \\
$2N$ & $\frac{1-\gamma }{N}$ & $0$ & $0$ & $...$ & $\xi _{H}$%
\end{tabular}%
\label{Tab:JointDistribution}
\end{table}

Here, $\xi _{L}=q-\sqrt{\left( N\mu -q\right) ^{2}+N\sigma ^{2}}$ is the required low realized value, $\xi
_{H}=q+\sqrt{\left( N\mu -q\right) ^{2}+N\sigma ^{2}}$ is the required high realized value, and $\gamma =\frac{1%
}{2}+\frac{N\mu -q}{2\sqrt{\left( N\mu -q\right) ^{2}+N\sigma ^{2}}}=\Pr(\xi=\xi_L)$ in accordance with Remark \ref{rem:equality} after recognizing that $\mathbb{E}(\xi)=N\mu$ and $Var(\xi)=N\sigma^2$. Each individual $X_n$ has \emph{three} realizations but the sum $\xi$ has only two realized values. Although the joint distribution in Table \ref{Tab:JointDistribution} satisfies all the moment conditions (such as mean, variance, and pair-wise co-variance), it violates the independent constraints. In general, the sum $\xi $ must have at least $\left( N+1\right) $ different
realized values, even when each i.i.d.~$X_{n}$ has only two
realized values. Thus, we can never design a distribution for iid
random variables that would make the bounds in Lemma \ref{L:N-NonSharp}
sharp. This underscores the challenge caused by independence.

\subsection{Duality Method and Independence Constraints}
A prevalent approach to developing moment-based concentration bounds in the
Management Science/Operations Research (MS/OR) community is to use the primal-dual
method \citep[e.g.,][]{JSmith1995, van2016generalized}. However, independent
constraints make this method ineffective. To illustrate the new challenges,
we consider the bivariate ex-post payoff $Z\left( X_{1}+X_{2}|q\right)
=(X_{1}+X_{2}-q)^{+}$ where $X_{1}$ and $X_{2}$ are iid. The ex-post payoff $%
Z(X_{1},X_{2}|q)$ is continuous in each $X_{n}$ for any given $q$ and hence,
\[
\sup_{F\left( \cdot ,\cdot \right) }\left\{ \int_{-\infty }^{\infty
}\int_{-\infty }^{\infty }Z\left( X_{1},X_{2}|q\right)
dF(x_{1},x_{2})\right\} =\max_{\lambda \left( X_{1},X_{2}\right)
}\sum_{X_{1}}\sum_{X_{2}}\lambda \left( X_{1},X_{2}\right) Z\left(
X_{1},X_{2}|q\right) ,
\]%
implying that it is sufficient to consider only (joint) discrete
distributions \citep[see Lemma 6.4 in][for details]{hettich1993semi}. 

Consider first the weaker assumption that the covariance between $X_{1}$ and $X_{2}$
is zero and formulate the following dual problem:%
\[
\begin{tabular}{lll}
& $D=\max_{\lambda \left( X_{1},X_{2}\right)
}\sum_{X_{1}}\sum_{X_{2}}\lambda \left( X_{1},X_{2}\right) Z\left(
X_{1},X_{2}|q\right) $ &  \\
$\text{s.t. }$ & $\sum_{X_{1}}\sum_{X_{2}}\lambda \left( X_{1},X_{2}\right)
=1$, & [Total probability] \\
& $\sum_{X_{1}}\sum_{X_{2}}\lambda \left( X_{1},X_{2}\right) X_{1}=\mu \text{%
, }$ & [Mean of $X_{1}$] \\
& $\sum_{X_{1}}\sum_{X_{2}}\lambda \left( X_{1},X_{2}\right) X_{2}=\mu \text{%
, }$ & [Mean of $X_{2}$] \\
& $\sum_{X_{1}}\sum_{X_{2}}\lambda \left( X_{1},X_{2}\right) X_{1}^{2}=\mu
^{2}+\sigma ^{2}\text{, }$ & [Variance of $X_{2}$] \\
& $\sum_{X_{1}}\sum_{X_{2}}\lambda \left( X_{1},X_{2}\right) X_{2}^{2}=\mu
^{2}+\sigma ^{2}\text{, }$ & [Variance of $X_{2}$] \\
& $\sum_{X_{1}}\sum_{X_{2}}\lambda \left( X_{1},X_{2}\right) X_{1}X_{2}=\mu
^{2}$, & [Covariance]%
\end{tabular}%
\]%
where $\lambda (X_{1},X_{2})$ is a general finite sequence such that $%
\lambda (X_{1},X_{2})$ is nonnegative but only a finite number of them can
be strictly positive. Let $y_{t}
$ be the shadow price of the moment constraints. Similar to equation (B.4) on page 26 in \cite{GQWWZ2023}, we formulate
the following primal problem:
\[
\begin{tabular}{ll}
& $P=\min_{y_{i}}\left\{ y_{0}+y_{1}\mu +y_{2}\mu +y_{3}\left( \mu
^{2}+\sigma ^{2}\right) +y_{4}\left( \mu ^{2}+\sigma ^{2}\right) +y_{5}\mu
^{2}\right\} $ \\
$\text{s.t. }$ & $%
y_{0}+y_{1}X_{1}+y_{2}X_{2}+y_{3}X_{1}^{2}+y_{4}X_{2}^{2}+y_{5}X_{1}X_{2}%
\geq Z\left( X_{1},X_{2}|q\right) ,\forall \left( X_{1},X_{2}\right),$%
\end{tabular}%
\]%
which is a linear semi-infinite programming problem with a finite number
of decision variables ($y_{t}$) but an infinite number of constraints. With both mean and variance being finite, we recognize that $%
\mathbb{E}[Z\left( X_1, X_2|q\right) ]$ is finite. With continuous $Z(X_1,X_2|q)$ and finite $%
\mathbb{E}[Z\left( X_1, X_2|q\right) ]$, we conclude that $P=D$ \citep[see
Lemma 6.5 in][for details]{hettich1993semi}. Technically, the primal model $P
$ is more tractable than the initial dual model $D$, making the primal-dual method
popular in the MS/OR literature. For example, based on the joint distribution in Table %
\ref{Tab:JointDistribution}, we can predict the binding constraints for
model $P$ and then derive the bound on $\mathbb{E}(X_{1}+X_{2}-q)^{+}$.

We now return to the independence constraints. Because we use the joint probability
mass $\lambda \left( X_{1},X_{2}\right) $ to formulate the dual problem $D$,
we face the following constraints%
\[
\begin{tabular}{ll}
$\lambda \left( X_{1},X_{2}\right) \lambda \left( X_{2},X_{1}\right)
-\lambda \left( X_{1},X_{1}\right) \lambda \left( X_{2},X_{2}\right) =0\text{%
, any }\left( X_{1},X_{2}\right) .$ & [Independence]%
\end{tabular}%
\]%
These constraints are \emph{nonlinear} with
respect to the joint probability mass $\lambda \left( X_{1},X_{2}\right) $
and the number of the independence constraints is infinite since there are
infinite pairs of $(X_{1},X_{2})$. These constraints can never
be \textquotedblleft dualized" into a linear semi-infinite programming
model. We can still let $y_{X_{1},X_{2}}$ be the shadow price of the independence
constraints associated with the given pair $\left( X_{1},X_{2}\right) $. Because the right-hand-side of a given independence constraint is now $0$, the objective function of $P$ will remain intact as $%
y_{X_{1},X_{2}}\cdot 0=0$. However, the left-hand-side of the independence constraint involves two products of four probability masses and must be
nonlinear. This observation implies that i) we lose the linear
semi-infinite characteristic (as the number of decision variables $%
y_{X_{1},X_{2}}$ will grow to infinite) and more importantly, ii) the
independence constraints make the dual model $D$ nonlinear. 

Alternatively, if we formulate the dual problem using the marginal
probability mass, then we face the following problem:
\begin{eqnarray*}
&&D=\max_{\lambda \left( \cdot \right) }\sum_{X_{1}}\sum_{X_{2}}\lambda
\left( X_{1}\right) \lambda \left( X_{2}\right) Z\left( X_{1},X_{2}|q\right)
\\
\text{s.t. } &&\sum_{X}\lambda \left( X\right) =1\text{, }\sum_{X}\lambda
\left( X\right) X=\mu \text{, and }\sum_{X}\lambda \left( X\right) X^{2}=\mu
^{2}+\sigma ^{2}.
\end{eqnarray*}%
Although the number of shadow prices in the new model based on marginal probability mass comes down to
six, the objective function becomes \emph{nonlinear} with respect to the
marginal probability mass $\lambda \left( \cdot \right) $, making the
corresponding primal model much less tractable than before. 

In summary, the independence constraints hinder the application of the
popular primal-dual method, prompting us to consider other approaches to
developing the bounds on tail probability and linear loss. 
From the theoretical perspective, many distribution-dependent models (such as the
bundle pricing and inventory models in Section \ref{S:Application}
assume independence). The extant MS/OR literature uses the aggregation method to develop
corresponding bounds without including independence explicitly; and these bounds are often
regarded as conservative \cite[see, e.g.,][]{van2016generalized}. We improve these conservative
bounds by adding independence. These new bounds are more suitable when
benchmarking the distribution-dependent models with independence.

\section{Tail Probability\label{S:Tail}}

\subsection{Equal Mean and Variance}

We can now present the first new result as follows.

\begin{proposition}
\label{P:Tail}(Tail Probability) When $\mathbb{E}\left( X _{n}\right) =\mu >%
\frac{q}{N}$ and $Var\left( X _{n}\right) =\sigma ^{2}>0$, it holds that%
\begin{equation}
\Pr \left( \xi =X _{1}+X _{2}+\ldots +X _{N}>q\right) \geq 1-\allowbreak 
\frac{\sigma ^{2N}}{\left( \left( \mu -\frac{q}{N}\right) ^{2}+\sigma
^{2}\right) ^{N}}.  \label{E:RightTail-N}
\end{equation}
\end{proposition}

\begin{proof}
Because each $X _{n}$ is independent, we find that%
\begin{equation*}
\mathbb{E}\left( e^{t(\xi-q) }\right) =\mathbb{E}\left[ e^{t\left( X _{1}-%
\frac{q}{N}\right) }\right]\mathbb{E}\left[e^{t\left( X _{2}-\frac{q}{N}%
\right) }\right]\ldots \mathbb{E}\left[e^{t\left( X _{N}-\frac{q}{N}\right) }%
\right] =\left[\mathbb{E}\left( e^{t\left( X _{n}-\frac{q}{N}\right) }\right)%
\right] ^{N}.
\end{equation*}%
By Lemma \ref{L:N=1}(a), $\Pr \left( X _{n}\leq \frac{q}{N}\right) \leq 
\frac{\sigma ^{2}}{\left( \mu -\frac{q}{N}\right) ^{2}+\sigma ^{2}}$, where
the equality sign holds for the extreme distribution $\Pr \left( X _{n}=%
\frac{q}{N}\right) =\frac{\sigma ^{2}}{\left( \mu -\frac{q}{N}\right)
^{2}+\sigma ^{2}}$ and $\Pr \left( X _{n}=\mu +\frac{\sigma ^{2}}{\left( \mu
-\frac{q}{N}\right) }\right) =\frac{\left( \mu -\frac{q}{N}\right) ^{2}}{%
\left( \mu -\frac{q}{N}\right) ^{2}+\sigma ^{2}}$. For exposition
simplicity, let 
\begin{equation*}
R=\mu +\frac{\sigma ^{2}}{\left( \mu -\frac{q}{N}\right) }-\frac{q}{N}=\frac{%
\left( \mu -\frac{q}{N}\right) ^{2}+\sigma ^{2}}{\mu -\frac{q}{N}}
\end{equation*}
be the range of the extreme distribution (i.e., the maximum minus the
minimum realized value).

We apply Markov's inequality to this extreme distribution to obtain 
\begin{eqnarray*}
\Pr \left( \xi -q>0\right) =\Pr \left( \left[ e^{t\left( X _{n}-\frac{q}{N}%
\right) }\right] ^{N}>1\right) \leq \left[\mathbb{E}\left( e^{t\left( X _{n}-%
\frac{q}{N}\right) }\right)\right] ^{N} \\
=\left[ \frac{\sigma ^{2}}{\left( \mu -\frac{q}{N}\right) ^{2}+\sigma ^{2}}+%
\frac{\left( \mu -\frac{q}{N}\right) ^{2}}{\left( \mu -\frac{q}{N}\right)
^{2}+\sigma ^{2}}e^{tR}\right] ^{N},
\end{eqnarray*}%
if $t>0$. Similarly, if $t<0$ then we have 
\begin{eqnarray*}
\Pr \left( \xi -q>0\right) &=&\Pr \left( \left[ e^{t\left( X _{n}-\frac{q}{N}%
\right) }\right] ^{N}<1\right) =1-\Pr \left( \left[ e^{t\left( X _{n}-\frac{q%
}{N}\right) }\right] ^{N}>1\right) \\
&\geq &1-\left[\mathbb{E}\left( e^{t\left( X _{n}-\frac{q}{N}\right) }\right)%
\right] ^{N}=1-\left[ \frac{\sigma ^{2}}{\left( \mu -\frac{q}{N}\right)
^{2}+\sigma ^{2}}+\frac{\left( \mu -\frac{q}{N}\right) ^{2}}{\left( \mu -%
\frac{q}{N}\right) ^{2}+\sigma ^{2}}e^{tR}\right] ^{N}.
\end{eqnarray*}%
As $t\rightarrow -\infty $, the second term in the bracket converges to
zero, yielding inequality (\ref{E:RightTail-N}).
\end{proof}


With independent distributions, the bound on tail probability is a rescaling
of the single-variable result (i.e., an analogy to Cram\'{e}r's Theorem).
When allocating the total \textquotedblleft budget" of $q=q_{1}+q_{2}+\dots
+q_{N}$, we simply let $q_{n}=\frac{q}{N}$ when random variables are iid. We
observe that when rescaling the single-variable result, we apply division to
the additive relationship (such as the total budget) and multiplication to
the probability of multi-fold convolution.

\subsection{Extensions}

As a natural extension, with non-equal mean and variance across different
random variables in the sum, we need to optimize the budget allocation.

\begin{proposition}
\label{P:EqualRange}(Equal Range) When the mean and variance of each $X _{n} 
$ are non-identical, the extreme distribution for each independent $X _{n}$
must have equal range.
\end{proposition}

\begin{proof}
Because the event $\xi >q$ is equivalent to the event $\sum_{n=1}^{N}\left(
X _{n}-q_{n}\right) >0$, similar to the proof of Proposition \ref{P:Tail},
we find that if $t<0$, then%
\begin{eqnarray*}
\Pr \left( \xi -q>0\right) &=&\Pr \left( \prod_{n=1}^{N}e^{t\left( X
_{n}-q_{n}\right) }<1\right) =1-\Pr \left( \prod_{n=1}^{N}e^{t\left( X
_{n}-q_{n}\right) }>1\right) \\
&\geq &1-\mathbb{E}\left[ \prod_{n=1}^{N}e^{t\left( X _{n}-q_{n}\right) }%
\right]=1-\prod_{n=1}^{N}\left[ \frac{\sigma _{n}^{2}}{\left( \mu
_{n}-q_{n}\right) ^{2}+\sigma _{n}^{2}}+\frac{\left( \mu _{n}-q_{n}\right)
^{2}}{\left( \mu _{n}-q_{n}\right) ^{2}+\sigma _{n}^{2}}e^{tR_{n}}\right],
\end{eqnarray*}%
where 
\begin{equation*}
R_{n}=\mu _{n}+\frac{\sigma _{n}^{2}}{(\mu _{n}-q_{n})}-q_{n}=\frac{\left(
\mu _{n}-q_{n}\right) ^{2}+\sigma _{n}^{2}}{\mu _{n}-q_{n}}
\end{equation*}%
represents the range of the extreme distribution associated with $X _{n} $.
As $t\rightarrow -\infty $, we obtain the result that%
\begin{equation*}
\Pr \left( \xi -q>0\right) \geq 1-\max_{q_{1},q_{2},...,q_{n}}\left\{
\prod_{n=1}^{N}\left[ \frac{\sigma _{n}^{2}}{\left( \mu _{n}-q_{n}\right)
^{2}+\sigma _{n}^{2}}\right] \right\} ,
\end{equation*}%
where $q=q_{1}+q_{2}+\dots +q_{N}$.

By focusing on%
\begin{equation*}
B=\max_{q_{1},q_{2},...,q_{N}}\left\{ \prod_{n=1}^{N}\left( \frac{\sigma
_{n}^{2}}{(\mu _{n}-q_{n})^{2}+\sigma _{n}^{2}}\right) \right\} ,
\end{equation*}%
subject to the budget constraint $q=q_{1}+q_{2}+\dots +q_{N}$, we find that
the Lagrangian equals%
\begin{equation*}
\mathcal{L}=\prod_{n=1}^{N}\left( \frac{\sigma _{n}^{2}}{(\mu
_{n}-q_{n})^{2}+\sigma _{n}^{2}}\right) -\gamma \left( q_{1}+q_{2}+\dots
+q_{N}-q\right) .
\end{equation*}%
The first-order conditions require that%
\begin{eqnarray*}
\frac{\partial \mathcal{L}}{\partial q_{n}} &=&\frac{2\left( \mu
_{n}-q_{n}\right) \sigma _{n}^{2}}{\left( (\mu _{n}-q_{n})^{2}+\sigma
_{n}^{2}\right) ^{2}}\prod_{i\neq n}\left( \frac{\sigma _{i}^{2}}{(\mu
_{i}-q_{i})^{2}+\sigma _{i}^{2}}\right) -\gamma \\
&=&\frac{2\left( \mu _{n}-q_{n}\right) }{(\mu _{n}-q_{n})^{2}+\sigma _{n}^{2}%
}\prod_{i=1}^{N}\left( \frac{\sigma _{i}^{2}}{(\mu _{i}-q_{i})^{2}+\sigma
_{i}^{2}}\right) -\gamma =0,
\end{eqnarray*}%
where $\gamma $ is the Lagrangian multiplier associated with the budget
constraint. We find that for any $n\neq m$, 
\begin{equation*}
\frac{2}{R_{n}}\prod_{i=1}^{N}\left( \frac{\sigma _{i}^{2}}{(\mu
_{i}-q_{i})^{2}+\sigma _{i}^{2}}\right) =\gamma =\frac{2}{R_{m}}%
\prod_{i=1}^{N}\left( \frac{\sigma _{i}^{2}}{(\mu _{i}-q_{i})^{2}+\sigma
_{i}^{2}}\right) ,
\end{equation*}%
indicating that $R_{n}=R_{m}$ for any $n\neq m$.
\end{proof}

We refer to the result in Proposition \ref{P:EqualRange} as the \emph{equal
range} property. In the two-dimensional model, Proposition \ref{P:EqualRange}
implies that the extreme joint distribution graphically forms a square,
which has important implications in models of bundling using mixed
strategies (see Section \ref{S:Bundle} for details).

Due to symmetry, we can also obtain the tail probability of the other
direction as follows. If $\mu <\frac{q}{N}$, then%
\begin{equation}
\Pr (\xi <q)\geq 1-\left( \frac{\sigma ^{2}}{\left( \mu -\frac{q}{N}\right)
^{2}+\sigma ^{2}}\right) ^{N}.  \label{E:LeftTail-N}
\end{equation}%
Let $\gamma \in (0,1)$ and $q_{N}\left( \gamma \right) =F_{N}^{-1}\left(
\gamma \right) $, where $F_{N}$ is the $N$-fold convolution of $F$. We refer
to $q_{N}\left( \gamma \right) $ as the $100\times \gamma $-th percentile of
the $N$-fold convolution of $F$. Using inequalities (\ref{E:RightTail-N})
and (\ref{E:LeftTail-N}), we immediately obtain the range of the percentile
as follows.

\begin{corollary}
\label{C:Percentile} (Percentile) It holds that%
\begin{equation}
N\mu -N\sigma \sqrt{\frac{1-\gamma ^{\frac{1}{N}}}{\gamma ^{\frac{1}{N}}}}%
\leq q_{N}\left( \gamma \right) \leq N\mu +N\sigma \sqrt{\frac{1-(1-\gamma
)^{\frac{1}{N}}}{(1-\gamma )^{\frac{1}{N}}}}.  \label{E:RangePercentileN}
\end{equation}
\end{corollary}

\begin{proof}
According to inequality (\ref{E:RightTail-N}) and the definition of $%
q_{N}\left( \gamma \right) $, we find that $1-\allowbreak \frac{\sigma ^{2N}%
}{\left( \left( \mu -\frac{q}{N}\right) ^{2}+\sigma ^{2}\right) ^{N}}\geq
1-\gamma $. By taking the root less than $N\mu $, we obtain that $N\mu
-N\sigma \sqrt{\frac{1-\gamma ^{\frac{1}{N}}}{\gamma ^{\frac{1}{N}}}}\leq
q_{N}\left( \gamma \right) $. Similarly, we obtain the bound in the
different direction using the condition $1-\allowbreak \frac{\sigma ^{2N}}{%
\left( \left( \mu -\frac{q}{N}\right) ^{2}+\sigma ^{2}\right) ^{N}}\leq
\gamma $ and taking the root larger than $N\mu $.
\end{proof}

When $\gamma =0.5$ and $N=1$, inequality (\ref{E:RangePercentileN}) yields
the well-known result that the median is between $\mu -\sigma $ and $\mu
+\sigma $. When extending to $N\geq 2$, the median of the sum of $N$ iid
random variables satisfies 
\begin{equation*}
N\mu -N\sigma \sqrt{\frac{1-\left( 0.5\right) ^{\frac{1}{N}}}{\left(
0.5\right) ^{\frac{1}{N}}}}\leq q_{N}\left( 0.5\right) \leq N\mu +N\sigma 
\sqrt{\frac{1-\left( 0.5\right) ^{\frac{1}{N}}}{\left( 0.5\right) ^{\frac{1}{%
N}}}}.
\end{equation*}%
When $N$ approaches infinity, the correction factor $\sqrt{\frac{1-\left(
0.5\right) ^{\frac{1}{N}}}{\left( 0.5\right) ^{\frac{1}{N}}}}$ approaches
zero, suggesting that the sample average of the median, which equals $\frac{1%
}{N}q_{N}\left( \frac{1}{2}\right) $, approaches the mean $\mu $. Without
independence, the aggregation bounds predict that the median is between $%
N\mu -\sqrt{N}\sigma $ and $N\mu +\sqrt{N}\sigma $, where the correction
factor equals $\frac{1}{\sqrt{N}}$ and is larger than $\sqrt{\frac{1-\left(
0.5\right) ^{\frac{1}{N}}}{\left( 0.5\right) ^{\frac{1}{N}}}}$ .

\section{Expected Loss\label{S:Loss}}

We recenter each random variable by using $\mu =\mu ^{\prime }-\frac{q}{N}$
as the shifted mean so that we focus on the bound for $\mathbb{E}\left( \xi
\right) ^{+}$. Due to recentering, a positive (negative) realized value of $X
$ implies a realized value larger (smaller) than $\frac{q}{N}$.

\subsection{Single-Dimensional Case\label{SS:ScarfNew}}

Korkine's identity \citep[Ch.~9 in][on pp.~242-243]{MPF1993} pertains to
covariance as follows:%
\begin{equation*}
\mathbb{E}\left[ \left( X-\mathbb{E}(X)\right) \left( Y-\mathbb{E}(Y)\right) %
\right] =\frac{1}{2}\mathbb{E}\left[ \left( X-X^{\prime }\right) \left(
Y-Y^{\prime }\right) \right] ,
\end{equation*}%
where $\left( X,Y\right) $ are iid copies of $\left( X^{\prime },Y^{\prime
}\right) $. In the special case with $N=1$, we find that%
\begin{equation*}
\mathbb{E}\left( X\right) ^{+}=\mathbb{E}\left( X\right) \mathbb{E}\left( 
\mathbb{I}_{\{X>0\}}\right) +\frac{1}{2}\mathbb{E}_{X,X^{\prime }}\left[
\left( X-X^{\prime }\right) \left( \mathbb{I}_{\{X>0\}}-\mathbb{I}%
_{\{X^{\prime }>0\}}\right) \right] =\mu \left( 1-\beta \right) +\frac{1}{2}%
T,
\end{equation*}%
in which $X$ and $X^{\prime }$ are iid and $\beta \equiv \Pr \left( X\leq
0\right) $ is a known probability based on a given feasible distribution.
With $\mu \left( 1-\beta \right) $ being fixed, we wish to maximize the term 
$\frac{1}{2}T$, which equals the covariance between variable $X$ and
indicator $\mathbb{I}_{\{X>0\}}$. The integrand $A_{1}$ in $T$ equals%
\begin{equation*}
A_{1}\equiv \left( X-X^{\prime }\right) \left( \mathbb{I}_{\{X>0\}}-\mathbb{I%
}_{\{X^{\prime }>0\}}\right) =\left\vert X-X^{\prime }\right\vert \left( 
\mathbb{I}_{\{\max \left\{ X,X^{\prime }\right\} >0\}}-\mathbb{I}_{\{\min
\left\{ X,X^{\prime }\right\} >0\}}\right) ,
\end{equation*}%
where the subscript $1$ indicates a one-dimensional model.

We observe that the indicator is weakly increasing in $X$. When $\left\vert
X-X^{\prime }\right\vert $ increases, the coefficient $\left( \mathbb{I}%
_{\{\max \left\{ X,X^{\prime }\right\} >0\}}-\mathbb{I}_{\{\min \left\{
X,X^{\prime }\right\} >0\}}\right) $ weakly increases. The integrand $A_{1}$
must be zero if both $X$ and $X^{\prime }$ have the same sign. Therefore,
the extreme distribution maximizing the summation $T$ must be a two-point
distribution with one positive realized value and one negative realized
value. According to the mean-variance conditions and probability constraint $%
\beta =\Pr \left( X\leq 0\right) $, we find that the extreme distribution is
unique and satisfies:%
\begin{equation}
\left\{ 
\begin{array}{l}
\Pr \left( X=\mu -\sigma \sqrt{\frac{1-\beta }{\beta }}\overset{def}{=}%
L\right) =\beta , \\ 
\Pr \left( X=\mu +\sigma \sqrt{\frac{\beta }{1-\beta }}\overset{def}{=}%
H\right) =1-\beta .%
\end{array}%
\right.  \label{E:F-2Point}
\end{equation}
The range of the two-point distribution in equation (\ref{E:F-2Point})
equals $H-L=\sigma \sqrt{\beta \left( 1-\beta \right) }$. Inequality (\ref%
{E:RangePercentileN}) implies that $L=inf_F{F^{-1}(\beta)}$ and $H=sup_F{%
F^{-1}(\beta)}$, where $F$ satisfies the mean-variance condition.

\begin{lemma}
\label{L:TighterScarf}When $\mathbb{E}(X)=\mu$, $Var(X)=\sigma^2>0$, and $\Pr(X\leq 0)=\beta$, it holds that%
\begin{equation}
\mathbb{E}\left( X\right) ^{+}\leq \mu \left( 1-\beta \right) +\sigma \sqrt{%
\beta \left( 1-\beta \right) } = (1-\beta) H.  \label{E:BoundBetaLoss}
\end{equation}
\end{lemma}

Importantly, the bound in Lemma \ref{L:TighterScarf} is tighter than Scarf's
bound because of the probability constraint $\beta =\Pr \left( X\leq
0\right) $. If we optimize the bound in inequality (\ref{E:BoundBetaLoss})
by choosing $\beta $, we recover Scarf's bound because the first order
condition yields that $\beta ^{\ast }=\frac{1}{2}\pm \frac{\mu -q}{2\sqrt{%
\left( \mu -q\right) ^{2}+\sigma ^{2}}}$ (since $\mu =\mu ^{\prime }-q$ is
the shifted mean). The advantage of using Lemma \ref{L:TighterScarf} is that
via the input variable $\beta $, we can now consider what is known as the
service level requirement, which is often linked to the tail probability %
\citep[see, e.g.,][p.~79]{A2000Book}, whereas the standard Scarf model does
not allow us to do that. \cite{PIJ2004} proved Lemma \ref{L:TighterScarf}
using Holder's inequality whereas we use the unique extreme distribution to
directly compute the bound. This difference in the proof becomes crucial
when developing the bound for the multi-dimensional model.

\begin{remark}
\label{R:ConditionalMeans}When $\mathbb{E}(X)=\mu$, $Var(X)=\sigma^2>0$, and $\Pr(X\leq 0)=\beta$, it holds that $\mathbb{E}\left( X\right) ^{-}\geq
\mu \beta -\sigma \sqrt{\beta \left( 1-\beta \right) }=\beta L$.
Consequently, $\mathbb{E}\left( X|X>0\right) \leq H$ and $\mathbb{E}\left(
X|X\leq 0\right) \geq L$, where $H$ and $L$ are shown in equation (\ref%
{E:F-2Point}).
\end{remark}

We refer to $\mathbb{E}\left( X|X\leq 0\right) $ and $\mathbb{E}\left(
X|X>0\right) $ as the left and right conditional means of the iid random
variable $X$, respectively. The bounds on the conditional means will play an
important role in the subsequent analysis. We now highlight a crucial
difference between the one-dimensional and multi-dimensional models. Let $%
\beta =\Pr \left( X_{n}\leq 0\right) $ be the first input and $\gamma =\Pr
\left( \xi \leq 0\right) $ be the second input. A notable relationship is
that%
\begin{equation}
1-\beta ^{N}\geq 1-\gamma \geq \left( 1-\beta \right) ^{N}.
\label{E:TwoProbRelate}
\end{equation}%
The first inequality indicates that the event that all $X_{n}$ are
non-positive must imply the event that the sum is non-positive, but the
opposite is not true. The second inequality indicates that the event that
all $X_{n}$ are positive must imply the event that the sum is positive, but
the opposite is not true. Inequalities (\ref{E:TwoProbRelate}) hold for any
iid distributions. In the one-dimensional model, $\gamma =\beta $ must hold
due to only one dimension; whereas in the multi-dimensional model, the
one-to-one mapping between $\gamma $ and $\beta $ does not exist. Only when
one of these constraints becomes binding, do we re-establish the one-to-one
mapping.

\subsection{Two-Point Distributions}

Several known inequalities in the literature show that the extreme
distributions come from the family of two-point distributions. For instance, 
\cite{Bentkus2004} proved that $\Pr (\xi \geq q)\leq c\Pr
(s_{1}+s_{2}+...+s_{N}\geq q)$, where $\xi $ is a sum of $N$ independent
bounded random variables, $c$ is a constant and each $s$ is iid Bernoulli,
while \cite{M2003MAD} developed bounds based on mean and absolute deviations
using Binomial distributions. Motivated by the literature, we first compute
a candidate bound based on two-point distributions. In Section \ref%
{SS:TrueExtreme}, we prove that the candidate bound is indeed the globally
optimal bound among all feasible distributions subject to the mean-variance
condition.

\subsubsection{Piecewise Objective Function}

As any two-point distribution can be fully characterized by equation (\ref%
{E:F-2Point}), where $\beta =\Pr (X_n\leq 0)$ is the decision variable, the
sum $\xi $ follows a Binomial distribution satisfying the following
probability mass function:%
\begin{eqnarray}
&&\Pr \left( \xi =\left( N-t\right) L+tH\right)  \notag \\
&=&\Pr \left( \xi =\left( N-t\right) \left( \mu -\sigma \sqrt{\frac{1-\beta 
}{\beta }}\right) +t\left( \mu +\sigma \sqrt{\frac{\beta }{1-\beta }}\right)
\right) =\frac{N!}{t!(N-t)!}\beta ^{N-t}\left( 1-\beta \right) ^{t},
\label{E:EndogenousBinomial}
\end{eqnarray}%
where $t\in \left\{ 0,1,...,N\right\} $ is the number of times $X_n=H$. This
distribution is endogenous in the sense that it is generated using the
mean-variance constraint and the probability constraint satisfied by the
extreme marginal distributions.

In order to derive the optimal bound for the expected loss based the
endogenous distribution of $\xi$, we define a sequence of thresholds $%
\left\{ \delta _{k}\right\} $ satisfying%
\begin{equation}
0=N\mu +\sigma \left( -\left( N-k\right) \sqrt{\frac{1-\delta _{k}}{\delta
_{k}}}+k\sqrt{\frac{\delta _{k}}{1-\delta _{k}}}\right) ,
\label{E:Beta_kValue}
\end{equation}%
where $k=1,2,...,N-1$. In essence, $\delta_k$'s are the values of $\beta$
for which $\xi=0$ given $t$. By default, we let $\delta _{0}=1$ and $\delta
_{N}=0$ so that if $\beta \in \left[ \delta _{k},\delta _{k-1}\right] $,
then $\xi>0$ 
for $t\geq k$ and $\xi\le 0$ 
for $t\leq k-1$.

\begin{lemma}
\label{L:PieceWiseObj}(Piecewise Objective Function) Under the endogenous
Binomial distribution in equation (\ref{E:EndogenousBinomial}), the expected
loss equals 
\begin{equation}
Z\left( \beta \right) \equiv N\sigma \sqrt{\beta \left( 1-\beta \right) }%
\frac{\left( N-1\right) !\left( 1-\beta \right) ^{k-1}\beta ^{N-k}}{%
(k-1)!(N-k)!}+N\mu \sum_{t=k}^{N}\frac{N!\beta ^{N-t}\left( 1-\beta \right)
^{t}}{t!(N-t)!},  \label{E:ZbetaPiece}
\end{equation}%
for any $\beta \in \left[ \delta _{k},\delta _{k-1}\right] $.
\end{lemma}

This lemma provides interesting insights into the shape of the expected
loss. The second-order conditions reveal that each piece of $Z(\beta )$ is
concave, resulting in multiple local optima with respect to $\beta$. 
We illustrate the piecewise objective function $Z\left( \beta \right) $ in
Figure \ref{F:Piecewise} for two constellations of $(\mu, \sigma)$ at $N=5$.
An upper bound on the expected loss can be obtained by optimizing this
piecewise objective function.

\begin{figure}[tbph]
\caption{Piecewise Objective Function}
\label{F:Piecewise}%
\subfigure[$Z(\beta)$ under $\mu=0, \sigma=1, N=5$
]{\includegraphics[width=3.1in]{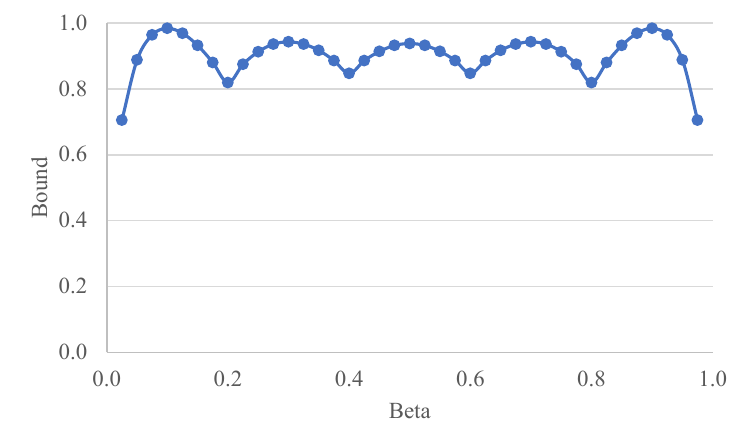}}\hspace{0.05in} 
\subfigure[$Z(\beta)$ under $\mu=-0.1, \sigma=1,
N=5$]{\includegraphics[width=3.1in]{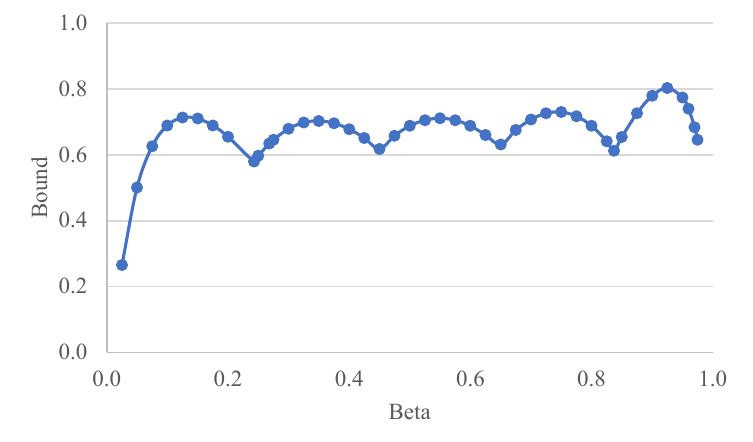}}
\end{figure}

\subsubsection{Zero Mean}

The case with zero mean (i.e., $\mu =\mu ^{\prime }-\frac{q}{N}=0$) provides
invaluable insights into the local optima. As a special case of equation (%
\ref{E:Beta_kValue}), we find that (i) the thresholds satisfy $\delta _{k}=%
\frac{N-k}{N}$, i.e., are decreasing in $k$, and (ii) the second term on the
right-hand side of equation (\ref{E:ZbetaPiece}) equals zero. Therefore, in
order to obtain the bound, we only need to maximize the first term on the
right-hand side of equation (\ref{E:ZbetaPiece}), which can be written as
follows: 
\begin{equation*}
T\left( \beta \right) \equiv T\left( k,\beta \right) \equiv N\sigma \sqrt{%
\beta \left( 1-\beta \right) }\frac{\left( N-1\right) !\left( 1-\beta
\right) ^{k-1}\beta ^{N-k}}{\left( k-1\right) !\left( N-k\right) !},
\end{equation*}%
where $\beta \in \left[ \frac{N-k}{N},\frac{N-k+1}{N}\right] $. Each piece $%
T\left( k,\beta \right) $ is continuous and strictly concave with respect to 
$\beta $. It is easy to see that the local optimal solution is $\beta
_{k}^{\ast }=\frac{2N-2k+1}{2N}$, which is the mid-point of the
corresponding interval $\left[ \frac{N-k}{N},\frac{N-k+1}{N}\right] $.
Substituting the local optimal solution into $T\left( k,\beta \right) $, we
obtain the local optimal objective values:%
\begin{equation}
T^{\ast }\left( k\right) \equiv T\left( k,\beta _{k}^{\ast }\right) =N\sigma%
\frac{\left( N-1\right) !\left( \frac{2k-1}{2N}\right) ^{k-1}\left( 1-\frac{%
2k-1}{2N}\right) ^{N-k}}{\left( k-1\right) !\left( N-k\right) !}\sqrt{\left( 
\frac{2k-1}{2N}\right) \left( \frac{2N-2k+1}{2N}\right) }.
\label{E:LocalOptimalValue}
\end{equation}

\begin{lemma}
\label{L:LocalOptimum}(Local Optima) If $\mu =0$ then the local optimal
objective values display the following properties: (i) Symmetric property
that $T^{\ast }\left( k\right) =T^{\ast }\left( N-k\right) $ holds for any $%
k=1,2,...,N$; (ii) Log-convexity with respect to $k$ such that $%
\max_{k}\left\{ T^{\ast }\left( k\right) \right\} =T^{\ast }\left( 1\right)
=T^{\ast }\left( N\right) $, where%
\begin{equation}
T^{\ast }\left( 1\right) =T^{\ast }\left( N\right) =N\sigma\left( 1-\frac{1}{%
2N}\right) ^{N-1}\sqrt{\frac{1}{2N}\left( 1-\frac{1}{2N}\right) }.
\label{E:TmaximumZeroM}
\end{equation}
\end{lemma}

As a direct consequence of Lemma \ref{L:LocalOptimum}, we find that with
zero mean, it holds that $T(\beta )\leq N\sigma\left( 1-\frac{1}{2N}\right)
^{N-1}\sqrt{\frac{1}{2N}\left( 1-\frac{1}{2N}\right) }$, and there exist two
extreme distributions attaining this bound, one with $\beta =\frac{1}{2N}$
and the other with $\beta =\frac{2N-1}{2N}$ . We illustrate the sequence of $%
T^{\ast }(k)$ in Figure \ref{F:LogConvex} using $N=5$ and $\sigma =1$. The
solid curve depicts the piecewise objective function while the dashed curve
connects all the local peaks as a log-convex curve.

\begin{figure}[tbph]
\caption{Log-Convexity}
\label{F:LogConvex}\centering \includegraphics[width=0.6%
\linewidth]{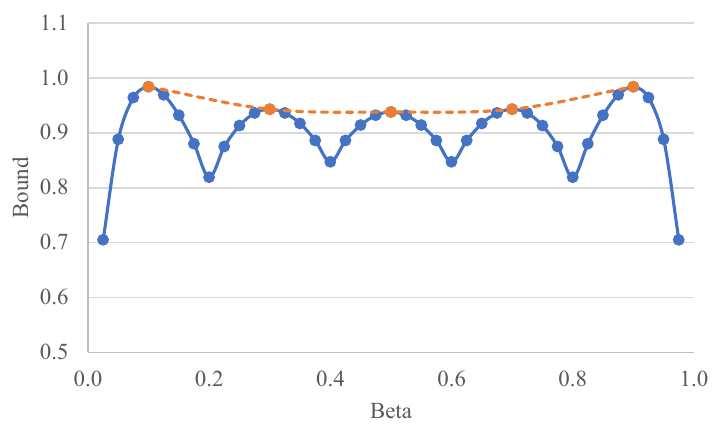}
\end{figure}

\subsubsection{Optimal Bound\label{SS:UltimateBounds}}

Equation (\ref{E:ZbetaPiece}) yields two noteworthy cases: (i) when $\beta
\in \left[ \delta _{1},1\right] $, we find that%
\begin{eqnarray*}
Z\left( \beta \right) &=&Z_{1}\left( \beta \right) \equiv N\mu \left(
\sum_{t=1}^{N}\left( 
\begin{array}{c}
N \\ 
t%
\end{array}%
\right) \beta ^{N-t}\left( 1-\beta \right) ^{t}\right) +T\left( 1,\beta
\right) \\
&=&N\mu \left( 1-\beta ^{N}\right) +N\sigma \sqrt{\beta \left( 1-\beta
\right) }\beta ^{N-1}.
\end{eqnarray*}%
since $k=1$, and (ii) when $\beta \in \left[ 0,\delta _{N-1}\right] $, we
find that%
\begin{eqnarray*}
Z\left( \beta \right) &=&Z_{N}\left( \beta \right) \equiv N\mu \left(
\sum_{t=N}^{N}\left( 
\begin{array}{c}
N \\ 
t%
\end{array}%
\right) \beta ^{N-t}\left( 1-\beta \right) ^{t}\right) +T\left( N,\beta
\right) \\
&=&N\mu \left( 1-\beta \right) ^{N}+N\sigma \sqrt{\beta \left( 1-\beta
\right) }\left( 1-\beta \right) ^{N-1},
\end{eqnarray*}%
since $k=N$. When optimizing the piecewise objective function $Z\left(
\beta \right) $, the optimal solution either falls in the rightmost interval 
$\left[ \delta _{1},1\right] $ or the leftmost interval $\left[ 0,\delta
_{N-1}\right] $ but never falls in between the two intervals. Thus, we
either optimize $Z_{1}(\beta )$ or $Z_{N}(\beta )$ to determine the optimal
bound. We summarize the results as follows.

\begin{proposition}
\label{P:ExpectedLoss}(Expected Loss) If $\mu <0$ then it holds that 
\begin{equation*}
Z^{\ast }\left( \beta \right) =N\mu +N\hat{\beta}_{1}^{N}\left( -\mu +\sigma 
\sqrt{\frac{1-\hat{\beta}_{1}}{\hat{\beta}_{1}}}\right) ,
\end{equation*}%
where%
\begin{equation}
\hat{\beta}_{1}=\frac{(2N-1)\sigma ^{2}+N\mu ^{2}-\mu \sqrt{\left(
2N-1\right) \sigma ^{2}+N^{2}\mu ^{2}}}{2N\left( \sigma ^{2}+\mu ^{2}\right) 
};  \label{E:BetaN1}
\end{equation}%
and if $\mu >0$ then it holds that%
\begin{equation*}
Z^{\ast }\left( \beta \right) =N\left( 1-\hat{\beta}_{2}\right) ^{N}\left(
\mu +\sigma \sqrt{\frac{\hat{\beta}_{2}}{1-\hat{\beta}_{2}}}\right) ,
\end{equation*}%
where%
\begin{equation}
\hat{\beta}_{2}=\frac{\sigma ^{2}+N\mu ^{2}-\mu \sqrt{\left( 2N-1\right)
\sigma ^{2}+N^{2}\mu ^{2}}}{2N\left( \sigma ^{2}+\mu ^{2}\right) }.
\label{E:BetaN2}
\end{equation}
\end{proposition}

Proposition \ref{P:ExpectedLoss} is valid due to the log-convexity of $%
T^{\ast }\left( k\right) $: when maximizing a log-convex objective function,
the optimal solution must be an extreme point. Interestingly, inequalities (%
\ref{E:TwoProbRelate}) imply two extreme points, corresponding to the
intervals $\left[ \delta _{1},1\right] $ and $\left[ 0,\delta _{N-1}\right] $%
. Thus, based on the sign of $\mu $ (i.e., whether $N\mu ^{\prime }>q$ or $%
N\mu ^{\prime }<q$), we choose either (\ref{E:BetaN1}) or (\ref{E:BetaN2})
to determine the candidate bound on the expected loss.

It is easy to extend this result to non-identical means and variances across 
$n$. Under the two-point distributions, 
we solve 
\begin{equation}
Z=\max_{\beta _{n}}\left\{ \sum_{n=1}^{N}\mu _{n}+\left( \Pi _{n=1}^{N}\beta
_{n}\right) \cdot \sum_{n=1}^{N}\left( -\mu _{n}+\sigma _{n}\sqrt{\frac{%
1-\beta _{n}}{\beta _{n}}}\right) \right\}  \label{E:NonEqualLoss}
\end{equation}%
to obtain the desired bound. As a result, we find that the extreme
distributions continue to display the \emph{equal range} property as $\sigma
_{n}\left( \sqrt{\frac{\beta _{n}^{\ast }}{1-\beta _{n}^{\ast }}}+\sqrt{%
\frac{1-\beta _{n}^{\ast }}{\beta _{n}^{\ast }}}\right) =R^{\ast }$ holds
for the optimal sequence $\beta _{n}^{\ast }$.

\subsection{Extreme Distribution\label{SS:TrueExtreme}}

A distinguishing feature of our derivation is that both the tail indicator
function $\mathbb{I}_{\{X>0\}}$ and linear loss $\max (0,\xi )$ have two
linear pieces, making two-point distributions the extreme distributions. We
can extend Korkine's identity to a multi-dimensional environment. Let $\xi
_{(i)}=X_{1}+...+X_{N}-X_{i}=\xi -X_{i}$ be the sum excluding the $i$-th
random variable and let $X_{i}^{\prime }$ be an independent copy of $X_{i}$,
both satisfying the mean-variance conditions. Also denote $\xi ^{\prime
}=\xi _{(i)}+X_{i}^{\prime }$, meaning that we keep the other $\left(
N-1\right) $ random variables intact but randomize the $i$-th random
variable one at a time. We observe that%
\begin{equation*}
\left( \xi -\xi ^{\prime }\right) \left( \mathbb{I}_{\{\xi >0\}}-\mathbb{I}%
_{\{\xi ^{\prime }>0\}}\right) =\left( X_{i}-X_{i}^{\prime }\right) \left( 
\mathbb{I}_{\{X_{i}+\xi _{(i)}>0\}}-\mathbb{I}_{\{X_{i}^{\prime }+\xi
_{(i)}>0\}}\right) .
\end{equation*}%
We define the component summation $T_{i}$ as follows:%
\begin{equation}
T_{i}=\mathbb{E}_{\left( X_{i},X_{i}^{\prime },\xi _{(i)}\right) }\left[
\left( X_{i}-X_{i}^{\prime }\right) \left( \mathbb{I}_{\{X_{i}+\xi
_{(i)}>0\}}-\mathbb{I}_{\{X_{i}^{\prime }+\xi _{(i)}>0\}}\right) \right] .
\label{E:ComponentSummation}
\end{equation}%
Due to symmetry caused by equal mean and variance, we obtain an intuitive
and important relationship as follows.

\begin{lemma}
\label{L:TotalComponent} (Total and Component Summations) The total
summation $T$ contains $N$ identical component summations, i.e., $T=NT_i$.
\end{lemma}

Lemma \ref{L:TotalComponent} implies that 
\begin{equation*}
\mathbb{E}\left( \xi \right) ^{+}=\mathbb{E}\left( \xi \right) \mathbb{E}%
\left( \mathbb{I}_{\{\xi >0\}}\right) +\frac{1}{2}T=N\mu \Pr \left( \xi
>0\right) +\frac{1}{2}NT_{i}.
\end{equation*}%
We define the integrand of the component summation as%
\begin{equation*}
A_{N}=\left\vert X_{i}-X_{i}^{\prime }\right\vert \left( \mathbb{I}_{\{\max
\left( X_{i},X_{i}^{\prime }\right) +\xi _{(i)}>0\}}-\mathbb{I}_{\{\min
\left( X_{i},X_{i}^{\prime }\right) +\xi _{(i)}>0\}}\right) ,
\end{equation*}%
where the subscript $N$ indicates the $N$-dimensional model. We now find that%
\begin{equation}
\mathbb{E}\left( \xi \right) ^{+}=N\mu \Pr \left( \xi >0\right) +\frac{N}{2}%
\sum_{X_{i}}\sum_{X_{i}^{\prime }}\sum_{\xi _{(i)}}A_{N}\Pr \left(
X_{i}\right) \Pr \left( X_{i}^{\prime }\right) \Pr \left( \xi _{(i)}\right) .
\label{E:BoundXi}
\end{equation}%
An important advantage of equation (\ref{E:BoundXi}) is that we can derive
the $Z(\beta )$ function with fewer steps (see the second proof of Lemma \ref%
{L:PieceWiseObj} in the Supplement). Equation (\ref{E:BoundXi}) also has other
future applications as it isolates the impact of each individual random
variable and bridges between the sum and variance (or absolute deviation) of
the random variables.

\begin{theorem}
\label{T:ExtremeDistributionT}(Extreme Distribution) When determining the
upper bound on $\mathbb{E}\left( \xi \right) ^{+}$, it suffices to consider
only the two-point distributions satisfying equation (\ref{E:F-2Point}).
\end{theorem}

To understand the intuition of Theorem \ref{T:ExtremeDistributionT}, we can
assume $X_{i}>0>X_{i}^{\prime }$ without loss of generality such that the
integrand $A_{N}$ increases along with the absolute value $\left\vert
X_{i}-X_{i}^{\prime }\right\vert $. The indicator function $\mathbb{I}%
_{\{X_{i}+\xi _{(i)}>0\}}$ weakly increases with respect to $X_{i}$ and $\xi
_{(i)}$. We find that $A_{N}$ must be increasing when $X_{i}^{\prime }$
decreases. Specifically, we find that (i) when $\xi _{(i)}\leq -X_{i}$, $%
A_{N}=0$ as both indicators are zero; (ii) when $-X_{i}<\xi _{(i)}\leq
-X_{i}^{\prime }$, $A=\left( X_{i}-X_{i}^{\prime }\right) $ as the first
indicator equals one but the second indicator equals zero; and (iii) when $%
-X_{i}^{\prime }<\xi _{(i)}$, $A=0$ as both indicators are one. Thus, to
increase the integrand, we increase $X_{i}$ but decrease $X_{i}^{\prime }$
as much as possible. By doing so, we also widen the interval $%
(-X_{i},-X_{i}^{\prime }]$, over which $A_{N}$ is strictly positive, making
the expected value $\mathbb{E}\left( A_{N}\right) $ even larger. Therefore,
the extreme distributions maximizing $\mathbb{E}\left( \xi \right) ^{+}$
must come from the family of two-point distributions. The candidate
solutions in Proposition \ref{P:ExpectedLoss} are indeed globally optimal. 

In Appendix, we provide a second proof without using Korkine's identity. The
second proof of Theorem \ref{T:ExtremeDistributionT} applies the bound for the quantile of the sum (i.e., Corollary \ref%
{C:Percentile}). If a candidate joint distribution for the sum yields a
quantile exceeding the range that inequality (\ref{E:RangePercentileN})
specifies, we can assert that this candidate joint distribution is
infeasible to the independent constraints. Thus, the bound on the quantile
of the sum can now act as a tractable replacement for the nonlinear
independent constraints. Subsequently, we obtain a relaxed objective value,
which notably can be attained by an independent sum of $N$ iid random variables that follow two-point distributions.

\subsection{Contrasting with Aggregate Bounds}

It is useful to contrast the bounds in Lemma \ref{L:N-NonSharp} with those
in Propositions\ \ref{P:Tail} and \ref{P:ExpectedLoss}. Figure \ref%
{F:TailPro}(a) evaluates the bounds $\frac{\sigma ^{2}}{N\left( \mu -\frac{q%
}{N}\right) ^{2}+\sigma ^{2}}\allowbreak $ and $\frac{\sigma ^{2N}}{\left(
\left( \mu -\frac{q}{N}\right) ^{2}+\sigma ^{2}\right) ^{N}}$ using the
parameters: $\mu =\sigma =1$ and $q=0.9$. Graphically, the former is higher
than the latter; and the gap between them can be visibly wide. Conceptually,
the former relaxes the independent constraints, providing an overestimate
(underestimate) for the left (right) tail. To highlight the speed of
convergence, we also depict the curves (in light grey) based on the normal
distribution. Figure \ref{F:TailPro}(a) confirms that relative to normal
prior, the bound produced by Proposition \ref{P:Tail} on the tail
probability is fairly accurate when $N$ increases (while that produced by
Lemma \ref{L:N-NonSharp} is much less accurate). With moderate $N$ ($2\leq
N\leq 7$), the gap between the new bound and normal prior remains
substantial, underscoring the most suitable parameter space for applying the
new bound on tail probability.

Using the parameters $\mu ^{\prime }=0=\mu -\frac{q}{N}$ and $\sigma =1$, we
find that $\mathbb{E}|\xi |=2\mathbb{E}(\xi )^{+}$. Figure \ref{F:TailPro}%
(b) contains plots of the aggregation bound $\mathbb{E}|\xi |=\sqrt{%
N^{2}\left( \mu -\frac{q}{N}\right) ^{2}+N\sigma ^{2}}$ and $2Z_{N}^{\ast }$%
. Graphically, the former bound is higher than the latter as the former
relaxes the independence constraints. In addition to visualization, we also
obtain several notable converging results. When each $X_{n}$ follows iid
standard normal distributions, we obtain that%
\begin{equation*}
\mathbb{E}(\xi )^{+}=\frac{1}{\sqrt{2\pi }}\sqrt{N}\approx 0.399\sqrt{N}=0.5%
\mathbb{E}|\xi |,
\end{equation*}%
where $\frac{1}{\sqrt{2\pi }}$ is the standard normal density evaluated at
point $x=0$. In contrast, equation (\ref{E:TmaximumZeroM}) yields that 
\begin{equation*}
\lim_{N\rightarrow \infty }\sqrt{N}\left( 1-\frac{1}{2N}\right) ^{N-1}\sqrt{%
\frac{1}{2N}\left( 1-\frac{1}{2N}\right) }=\frac{1}{\sqrt{2e}}\approx 0.429,
\end{equation*}%
indicating that the improved upper bound on expected loss equals $Z^{\ast
}=0.429\sqrt{N}$. Despite that the sum asymptotically converges to normal
distributions, the improved bound remains $7.5\%$ higher than the exact
value under standard normal prior. In contrast, the aggregation bound on
expected loss yields $\bar{Z}=0.5\sqrt{N}$, which is significantly higher
than $0.429\sqrt{N}$.

\begin{figure}[tbph]
\caption{Contrasting with Benchmark}
\label{F:TailPro}\center%
\subfigure[Tail Probability
]{\includegraphics[width=3.1in]{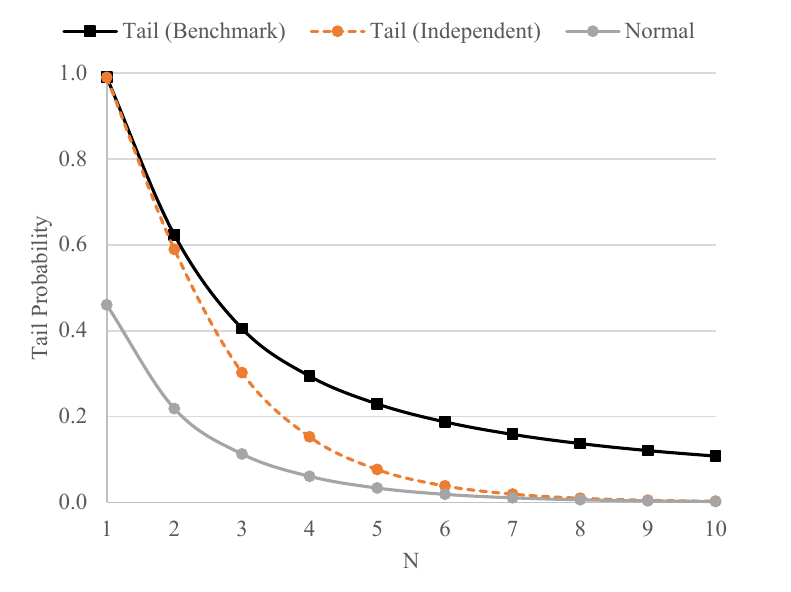}}\hspace{0.05in} 
\subfigure[Expected Linear
Loss]{\includegraphics[width=3.1in]{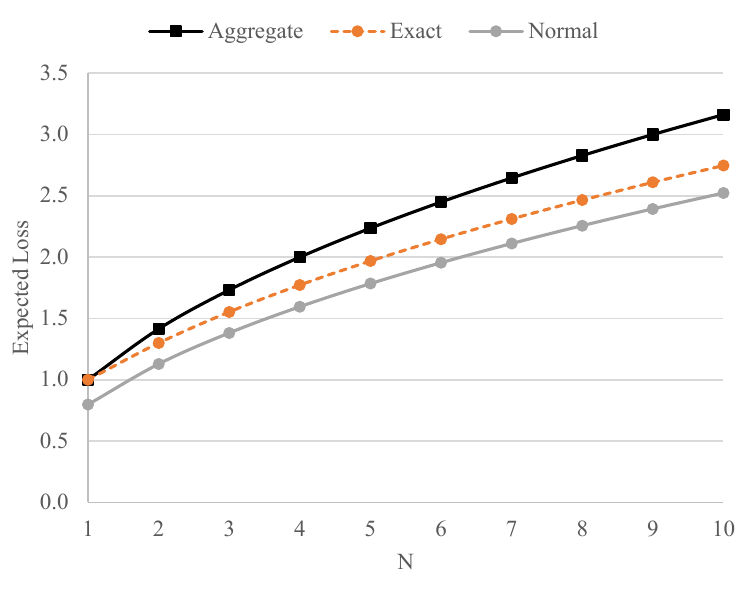}}
\end{figure}

\section{Applications\label{S:Application}}

\subsection{Bundle Pricing\label{S:Bundle}}

In the bundle pricing problem that \cite{CHP2022} studied, the firm selling $%
N$ goods in a bundle chooses a posted price $q$ and the customer's valuation
for good $n$ offered in the bundle is $X_{n}$.

\subsubsection{Equal Mean and Variance}

We first assume that $X_n$ has the same mean $\mu$ and same variance $%
\sigma^2$ for any $n$. To ensure that the firm's ex-post payoff is lower
semicontinuous, we assume that only when $X>q$, the customer buys the good;
otherwise, the customer walks away. With lower semicontinuity, the bound in
Proposition \ref{P:Tail} is attained rather than approached. With identical
mean and variance, the firm solves the following model:%
\begin{equation*}
Z=\max_{q}\left\{ q-\frac{q\sigma ^{2N}}{\left( \left( \mu -\frac{q}{N}%
\right) ^{2}+\sigma ^{2}\right) ^{N}}\right\} .
\end{equation*}

\begin{corollary}
\label{C:Bundle}(Pure Bundle Price) Let $t_{N}^{\ast }$ be the root of the
polynomial equation: 
\begin{equation}
1-2Nt^{2}-\left( t^{2}+1\right) ^{N}+t^{2}+2Nt\frac{\mu }{\sigma }%
-\allowbreak t^{2}\left( t^{2}+1\right) ^{N}=0.  \label{E:BundleN}
\end{equation}%
Then, the firm's optimal bundle price is $q_{N}^{\ast }=N\left( \mu
-t_{N}^{\ast }\sigma \right) $.
\end{corollary}

\begin{proof}
In order to define the extreme distribution, we introduce the safety factor $%
t$ as follows:%
\begin{equation*}
\left\{ 
\begin{array}{l}
\Pr \left( \tilde{X}_{n}=\mu -t\sigma \right) =\frac{1}{1+t^{2}}\equiv \beta
, \\ 
\Pr \left( \tilde{X}_{n}=\mu +\frac{1}{t}\sigma \right) =\frac{t^{2}}{1+t^{2}%
}=1-\beta ,%
\end{array}%
\right.
\end{equation*}%
for each good. Then, for $N\geq 1$, the bundle price is $q=N\left( \mu
-t\sigma \right) $, yielding the expected profit for the bundle:%
\begin{equation}
Z_{B}\equiv N(\mu -\sigma t)\left( 1-\frac{1}{\left( t^{2}+1\right) ^{N}}%
\right) .  \label{E:ZnGood}
\end{equation}%
When using a component pricing strategy, each product yields the expected
profit $Z_{n}=(\mu -\sigma t)\left( 1-\frac{1}{t^{2}+1}\right) $. Now, it
must hold that%
\begin{equation*}
Z_{B}=N(\mu -\sigma t)\left( 1-\frac{1}{\left( t^{2}+1\right) ^{N}}\right)
\geq N(\mu -\sigma t)\left( 1-\frac{1}{t^{2}+1}\right) =NZ_{n},
\end{equation*}%
implying that pure bundling is always better than component pricing. To
determine the optimal bundle price under independence, we take the first
derivative of equation (\ref{E:ZnGood}) with respect to $t$ as follows:%
\begin{equation*}
\frac{\partial Z_B}{\partial t} =\frac{N\sigma }{\left( t^{2}+1\right) ^{n}}%
-N\sigma +2N^{2}t\frac{\mu }{\left( t^{2}+1\right) ^{N+1}}-2N^{2}t^{2}\frac{%
\sigma }{\left( t^{2}+1\right) ^{N+1}}=0,
\end{equation*}%
which is equivalent to%
\begin{equation*}
N\sigma \left( 1-2Nt^{2}-\left( t^{2}+1\right) ^{N}+t^{2}+2Nt\frac{\mu }{%
\sigma }-\allowbreak t^{2}\left( t^{2}+1\right) ^{N}\right) =0.
\end{equation*}%
We thus confirm (\ref{E:BundleN}).
\end{proof}

In contrast, if we apply the bound based on aggregation from Lemma \ref%
{L:N-NonSharp}, we solve%
\begin{equation*}
\tilde{Z}=\max_{q}\left\{ q-\frac{q\sigma ^{2}}{N\left( \mu -\frac{q}{N}%
\right) ^{2}+\sigma ^{2}}\right\} .
\end{equation*}%
This yields the bundle price $\tilde{q}_{N}=N\mu -t^{\#}\sqrt{N}\sigma $,
where $t_{N}^{\#}$ solves a cubic equation $t^{3}+3t=\frac{2N\mu }{\sqrt{N}%
\sigma }=\frac{2\sqrt{N}\mu }{\sigma }$.

\begin{table}[tbph]
\caption{Bundle Pricing Solutions}
\label{Tab:Bundle}\centering%
\begin{tabular}{c|cccccc}
$N$ & $\tilde{q}_{N}$ & $\tilde{Z}_{B}$ & $q_{N}^{\ast }$ & $Z_{B}$ & $\frac{%
q_{N}^{\ast }-\tilde{q}_{N}}{q_{N}^{\ast }}$ & $\frac{Z_{B}-\tilde{Z}_{B}}{%
Z_{B}}$ \\ \hline
$1$ & \multicolumn{1}{r}{$1.346$} & \multicolumn{1}{r}{$0.769$} & 
\multicolumn{1}{r}{$1.346$} & \multicolumn{1}{r}{$0.769$} & 
\multicolumn{1}{r}{$0.0\%$} & \multicolumn{1}{r}{$0.0\%$} \\ 
$2$ & \multicolumn{1}{r}{$3.000$} & \multicolumn{1}{r}{$2.000$} & 
\multicolumn{1}{r}{$3.000$} & \multicolumn{1}{r}{$2.250$} & 
\multicolumn{1}{r}{$0.0\%$} & \multicolumn{1}{r}{$11.1\%$} \\ 
$3$ & \multicolumn{1}{r}{$4.767$} & \multicolumn{1}{r}{$3.401$} & 
\multicolumn{1}{r}{$4.824$} & \multicolumn{1}{r}{$3.991$} & 
\multicolumn{1}{r}{$1.2\%$} & \multicolumn{1}{r}{$14.8\%$} \\ 
$4$ & \multicolumn{1}{r}{$6.602$} & \multicolumn{1}{r}{$4.903$} & 
\multicolumn{1}{r}{$6.749$} & \multicolumn{1}{r}{$5.861$} & 
\multicolumn{1}{r}{$2.2\%$} & \multicolumn{1}{r}{$16.3\%$} \\ 
$5$ & \multicolumn{1}{r}{$8.484$} & \multicolumn{1}{r}{$6.476$} & 
\multicolumn{1}{r}{$8.741$} & \multicolumn{1}{r}{$7.811$} & 
\multicolumn{1}{r}{$2.9\%$} & \multicolumn{1}{r}{$17.1\%$} \\ 
$10$ & \multicolumn{1}{r}{$18.311$} & \multicolumn{1}{r}{$14.966$} & 
\multicolumn{1}{r}{$19.234$} & \multicolumn{1}{r}{$18.144$} & 
\multicolumn{1}{r}{$4.8\%$} & \multicolumn{1}{r}{$17.5\%$} \\ 
$20$ & \multicolumn{1}{r}{$38.979$} & \multicolumn{1}{r}{$33.468$} & 
\multicolumn{1}{r}{$41.337$} & \multicolumn{1}{r}{$40.010$} & 
\multicolumn{1}{r}{$5.7\%$} & \multicolumn{1}{r}{$16.4\%$} \\ 
&  &  &  &  &  & 
\end{tabular}%
\end{table}

Let $\tilde{Z}_{B}$ denote the profit bound when bundle pricing is done
using aggregation. We contrast the two solutions for the bundle price $%
\left( \tilde{q}_{N},q_{N}^{\ast }\right) $ and the two respective optimal
objective values $\left( \tilde{Z}_{B}, Z_{B}\right) $ in Table \ref%
{Tab:Bundle}. We use the parameter values $\mu =2.5$ and $\sigma =1$. As $N$
increases, the gap between $\tilde{q}_{N}$ and $q_{N}^{\ast }$ widens and $%
\tilde{q}_{N}$ is consistently lower than $q_{N}^{\ast }$. The gap in
profits is also remarkable, and reaches 17.5\% of optimal profit level $Z_B$%
, or, equivalently, 21.2\% of the profit level obtained using the
aggregation solution. This underscores the fact that a full account of
independence has a nontrivial impact on the quality of the bundle pricing
solution.

\subsubsection{Unequal Means and Variances}

When mean and variance of customer valuations are non-identical across
different products, we apply Proposition \ref{P:EqualRange} to define%
\begin{equation*}
R=\frac{\left( \mu _{n}-q_{n}\right) ^{2}+\sigma _{n}^{2}}{\mu _{n}-q_{n}},%
\text{ for all }n\text{,}
\end{equation*}%
as the universal range of all the valuations. Using the same notation as
above, we can write the firm's objective function as follows%
\begin{equation}
Z=\max_{R,q_{n}}\left\{ \left( \sum_{n=1}^{N}q_{n}\right) \left[
1-\prod_{n=1}^{N}\left( \frac{\sigma _{n}^{2}}{\left( \mu _{n}-q_{n}\right)
^{2}+\sigma _{n}^{2}}\right) \right] \right\} .  \label{E:T-Unequal}
\end{equation}%
The constraint on the range $R$ ensures that the internal budget allocation
of $q_{n}$ maximizes the product of $\prod_{n=1}^{N}\left( \frac{\sigma
_{n}^{2}}{\left( \mu _{n}-q_{n}\right) ^{2}+\sigma _{n}^{2}}\right) $.
Additionally, the definition of $R$ implies that $q_{n}=-\frac{1}{2}R+\mu
_{n}\pm \frac{1}{2}\sqrt{R^{2}-4\sigma _{n}^{2}}$ such that $R\geq 2\sigma
_{n}$ must hold. If both roots of $q_{n}$ are positive, we take the smaller
root; if there is a negative root, we take the positive root. Due to this
inconvenience of having two possible roots, it is not recommendable to
replace all $q_{n}$ with one variable $R$ when solving (\ref{E:T-Unequal}).
Instead, it is advisable to choose $R$ endogenously by requiring that the range of each random variable equals $R$.

\subsubsection{Mixed Bundle Pricing}

In situations where the underlying distribution is known, mixed bundling
strategies will (weakly) outperform pure bundling strategies. Mixed bundling
occurs when product $n$ is offered at the same time both separately at price 
$q_{n}$ and in a complete bundle at bundle price $q_{b}$, which may or may
not be equal to the bundle price $q_{B}$ under pure bundling. It usually
holds that $q_{b}\leq \sum_{n=1}^{N}q_{n}$ in a mixed bundling strategy.
However, it turns out that in terms of the worst distribution, mixed
bundling strategy is as effective as pure bundling strategy.

\begin{corollary}
\label{C:BundleMixed}(Mixed Bundle) When the firm's objective is to maximize
the worst-case expected profit, the mixed bundling strategy is as effective
as the pure bundling strategy.
\end{corollary}

\begin{proof}
As the extreme distribution displays the equal range property, we can easily
verify that the event $X_{n}>q_{n}$ but $X_{n}+\xi _{(n)}=\xi \leq q_{b}$
occurs with zero probability (i.e., the customer finds that buying only
product $n$ gives her a higher utility than buying the bundle). Therefore,
the firm's worst-case expected profit under mixed bundling with bundle price 
$q_{b} $ is identical to that under pure bundling with bundle price $q_{B}$.
We conclude that mixed bundling strategies do not improve the firm's
worst-case expected profit when valuations for each good are independent.
\end{proof}

Due to Corollary \ref{C:BundleMixed}, we either apply Proposition \ref%
{C:Bundle} or Equation (\ref{E:T-Unequal}) to determine the pure bundling
price, depending on whether or not the mean and variance are identical
across $n$. These results stand in contrast to those obtained by \cite%
{E2010JEMS} and \cite{B2013MSBundle} who advocate mixed bundling under
uniform distributions. Under the uniform distribution and zero production
costs with two products ($N=2$), both \citeauthor{E2010JEMS} and %
\citeauthor{B2013MSBundle} show that (i) the optimal mixed bundling is to
charge $q_{n}=\frac{2}{3}$ for product $n$ and $q_{b}=\frac{4-\sqrt{2}}{3}$
for both products; and (ii) the pure bundling strategy is to charge $q_{B}=%
\sqrt{\frac{2}{3}}$ for the bundle (and the price for individual product is
set at $q_{n}=1$ so that no customer buys only product $n$).

These distribution-specific strategies break down under the corresponding
extreme distributions with the same mean and variance. For example, if we
use $\mu =0.5 $ and $\sigma =\sqrt{\frac{1}{12}}$ as inputs in our
semiparametric analysis, we find that under either mixed or pure bundling
strategy where the bundle price is the same, i.e., $q_{B}=q_{b}$, the firm's
most unfavorable distribution remains the same, making the mixed and pure
bundling strategies equally profitable. The intuition is that the joint
distribution forms a square, as suggested by Proposition \ref{P:EqualRange},
so that the event that a customer buys only product $n$ does not occur. 
The optimal pricing strategies of \cite{E2010JEMS} and \cite{B2013MSBundle}
can be shown to be suboptimal under extreme distributions. Specifically,
when $q_{b}=\frac{4-\sqrt{2}}{3}$, the firm's expected profit under the
extreme distribution equals $0.086$; and when $q_{B}=\sqrt{\frac{2}{3}}$,
the firm's expected profit under the extreme distribution equals $0.137$. In
contrast, when using the robust bundling strategy, the bundle price is $%
q^{\ast }=0.527$ and the firm's expected profit is $0.338$. Additionally,
when $q^{\ast }=0.527$ is used under the uniform distribution of valuations,
the firm's expected profit equals $0.454$. We can conclude that the
distribution-specific prices are too high while the robust bundle price
provides a much better guarantee. 
As an additional benefit, our method can easily scale to an arbitrary number
of products.

\subsection{Inventory Management}

Suppose that a firm that owns a central warehouse chooses an inventory level 
$q$ prior to receiving the realized demand $X_{n}$ from retailer $n$. Each
retailer is treated equally with the same understocking and overstocking
costs $b$ and $h$, respectively. Thus, the choice of $q$ is equivalent to
choosing a forecast for $\xi $ subject to a generalized linear scoring rule.
The ex-post loss function can be written as follows:%
\begin{equation*}
Z\left( q,\xi \right) =b\left( \xi -q\right) ^{+}+h\left( q-\xi \right) ^{+}=%
\frac{h+b}{2}|\xi -q|+\frac{b-h}{2}\left( \xi -q\right) .
\end{equation*}%
According to Proposition \ref{P:ExpectedLoss}, the firm solves the following
problem:%
\begin{equation*}
Z=\min_{q}\left\{ \frac{\left( b+h\right) }{2}\max_{\beta }\left[ N\mu
-q+2\beta ^{N}\left( q-N\mu \right) +2N\beta ^{N}\sigma \sqrt{\frac{1-\beta 
}{\beta }}\right] +\frac{\left( b-h\right) }{2}\left( N\mu -q\right)
\right\} ,
\end{equation*}%
which represents a zero-sum game between the firm and adverse nature. The
firm chooses $q$ to minimize the cost $T(\beta ,q)$ but adverse nature
chooses $\beta $ to maximize the cost.

\begin{corollary}
\label{P:Backorder}(Inventory Risk-Pooling) Let $\beta ^{\ast }=\left( \frac{%
b}{b+h}\right) ^{\frac{1}{N}}$. 
If $\frac{b}{b+h}\geq \frac{1}{2}$, then the firm's most unfavorable
distribution is the following two-point distribution:%
\begin{equation}
\left\{ 
\begin{array}{l}
\Pr \left( \tilde{X}=\mu -\sigma \sqrt{\frac{1-\beta ^{\ast }}{\beta ^{\ast }%
}}\right) =\beta ^{\ast }, \\ 
\Pr \left( \tilde{X}=\mu +\sigma \sqrt{\frac{\beta ^{\ast }}{1-\beta ^{\ast }%
}}\right) =1-\beta ^{\ast },%
\end{array}%
\right.  \label{E:F-Backorder}
\end{equation}%
and the firm's optimal inventory level equals:%
\begin{equation}
q^{\ast }=N\mu +\sigma \left( \frac{2\beta ^{\ast }-1}{2\sqrt{(1-\beta
^{\ast })\beta ^{\ast }}}-\left( N-1\right) \sqrt{\frac{1-\beta ^{\ast }}{%
\beta ^{\ast }}}\right),  \label{E:BaseStock}
\end{equation}%
and the firm's optimal objective value equals:%
\begin{equation}
Z^{\ast }=b\sigma N\sqrt{\frac{1-\beta ^{\ast }}{\beta ^{\ast }}}.
\label{E:VarlueOfGame}
\end{equation}
\end{corollary}

\begin{proof}
As $\frac{b}{b+h}\geq 0.5$, the firm orders more inventory than the
aggregate mean, giving rise to the case of $N\mu <q$. Using the identity
that $\left( \xi -q\right) ^{+}$ $=\frac{\xi -q}{2}+\frac{1}{2}\left\vert
\xi -q\right\vert $ and Proposition \ref{P:ExpectedLoss}, we find that the
payoff function equals%
\begin{equation*}
Z\left( \beta ,q\right) =\frac{\left( b+h\right) }{2}\left[ N\mu -q+2\beta
^{N}\left( q-N\mu \right) +2N\beta ^{N}\sigma \sqrt{\frac{1-\beta }{\beta }}%
\right] +\frac{\left( b-h\right) }{2}\left( N\mu -q\right) .
\end{equation*}%
We solve the first-order conditions to determine the saddle point:%
\begin{eqnarray*}
\frac{\partial Z}{\partial q} &=&\left( h+b\right) \beta ^{N}-b=0 \\
\frac{\partial Z}{\partial \beta } &=&\left( h+b\right) \left[ N\beta
^{N-1}\left( q-N\mu \right) +\frac{N\sigma \beta ^{N-1}\left( 1-2\beta
\right) }{2\sqrt{\left( 1-\beta \right) \beta }}+(N-1)N\beta ^{N-2}\sigma 
\sqrt{\left( 1-\beta \right) \beta }\right] =0
\end{eqnarray*}%
The first condition $\frac{\partial Z}{\partial q}=0$ immediately yields
Equation (\ref{E:F-Backorder}). With some algebra, we find that the second
condition $\frac{\partial Z}{\partial \beta }=0$, along with Equation (\ref%
{E:F-Backorder}), yields that%
\begin{equation*}
\left( q-N\mu \right) 2\sqrt{\left( 1-\beta \right) \beta }+\sigma \left(
1-2\beta \right) +2(N-1)\sigma \left( 1-\beta \right) =0.
\end{equation*}%
After rearranging the terms, we confirm Equation (\ref{E:BaseStock}).
Because $\frac{\partial ^{2}Z}{\partial q^{2}}=0$, it is easy to verify the
second-order conditions to confirm that the pair $(\beta ^{\ast },q^{\ast })$
constitutes a saddle point. Finally, substituting $\beta ^{\ast }$ and $%
q^{\ast }$ in $Z(\beta ,q)$, we find that the value of the zero-sum game
equals $Z^{\ast }=b\sigma N\sqrt{\frac{1-\beta ^{\ast }}{\beta ^{\ast }}}$.
\end{proof}

To connect this result to the forecasting literature in econometrics, we
note that the linear loss constitutes a strictly proper scoring rule for
forecasting quantiles \citep[see, e.g.,][Theorem 6]{GR2007JaSa}. By
maximizing the expected score, a forecaster makes an honest forecast, which
is why proper scoring rules are widely used for measuring out-of-sample
forecast performance in many applications, e.g., the check loss function in
financial econometrics. Similarly, the optimal inventory level $q^{\ast}$
shown in Equation (\ref{E:BaseStock}) can be viewed as the robust optimal
forecast using an asymmetric piecewise linear scoring rule.

If we apply the non-sharp bound from Lemma \ref{L:N-NonSharp}, then we solve 
\begin{equation*}
\tilde{Z}=\min_{q}\left\{ \frac{\left( b+h\right) }{2}\sqrt{N^{2}\left( \mu -%
\frac{q}{N}\right) ^{2}+N\sigma ^{2}}+\frac{\left( b-h\right) }{2}\left(
N\mu -q\right) \right\} ,
\end{equation*}%
and find $\tilde{q}=N\mu +\frac{\sqrt{N}\sigma }{2}\left( \sqrt{\frac{b}{h}}-%
\sqrt{\frac{h}{b}}\right) $ and $\tilde{Z}=\sigma \sqrt{Nbh}$. In Table \ref%
{Tab:Inventory}, we contrast the two solutions $\left( \tilde{q},q^{\ast
}\right) $ and the optimal costs $\left( \tilde{Z},Z^{\ast }\right) $ using
the parameter values $\mu =2.5$, $\sigma =1$, $b=4$, and $h=1$. It can be
seen from the table that as $N$ increases, the gap between $\tilde{q}$ and $%
q^{\ast }$ remains moderate, with $\tilde{q}$ consistently exceeding $%
q^{\ast }$. When compared in terms of expected costs, the gap is more
visible and reaches $5.6\%$ of the optimal cost level. This underscores the
non-trivial effect that the independence constraint has on the quality of
solutions.

\begin{table}[tbph]
\caption{Inventory Solutions}
\label{Tab:Inventory}\centering%
\begin{tabular}{c|cccccc}
$N$ & $\tilde{q}$ & $\tilde{Z}$ & $q^{\ast }$ & $Z^{\ast }$ & $\frac{\tilde{q%
}-q^{\ast }}{q^{\ast }}$ & $\frac{\tilde{Z}-Z^{\ast }}{Z^{\ast }}$ \\ \hline
$1$ & \multicolumn{1}{r}{$3.250$} & \multicolumn{1}{r}{$2.000$} & 
\multicolumn{1}{r}{$3.250$} & \multicolumn{1}{r}{$2.000$} & 
\multicolumn{1}{r}{$0.0\%$} & \multicolumn{1}{r}{$0.0\%$} \\ 
$2$ & \multicolumn{1}{r}{$6.061$} & \multicolumn{1}{r}{$2.828$} & 
\multicolumn{1}{r}{$5.940$} & \multicolumn{1}{r}{$2.748$} & 
\multicolumn{1}{r}{$2.0\%$} & \multicolumn{1}{r}{$2.9\%$} \\ 
$3$ & \multicolumn{1}{r}{$8.799$} & \multicolumn{1}{r}{$3.464$} & 
\multicolumn{1}{r}{$8.605$} & \multicolumn{1}{r}{$3.335$} & 
\multicolumn{1}{r}{$2.3\%$} & \multicolumn{1}{r}{$3.9\%$} \\ 
$4$ & \multicolumn{1}{r}{$11.500$} & \multicolumn{1}{r}{$4.000$} & 
\multicolumn{1}{r}{$11.249$} & \multicolumn{1}{r}{$3.832$} & 
\multicolumn{1}{r}{$2.2\%$} & \multicolumn{1}{r}{$4.4\%$} \\ 
$5$ & \multicolumn{1}{r}{$14.177$} & \multicolumn{1}{r}{$4.472$} & 
\multicolumn{1}{r}{$13.879$} & \multicolumn{1}{r}{$4.273$} & 
\multicolumn{1}{r}{$2.1\%$} & \multicolumn{1}{r}{$4.7\%$} \\ 
$10$ & \multicolumn{1}{r}{$27.372$} & \multicolumn{1}{r}{$6.325$} & 
\multicolumn{1}{r}{$26.901$} & \multicolumn{1}{r}{$6.009$} & 
\multicolumn{1}{r}{$1.7\%$} & \multicolumn{1}{r}{$5.3\%$} \\ 
$20$ & \multicolumn{1}{r}{$53.354$} & \multicolumn{1}{r}{$8.944$} & 
\multicolumn{1}{r}{$52.655$} & \multicolumn{1}{r}{$8.474$} & 
\multicolumn{1}{r}{$1.3\%$} & \multicolumn{1}{r}{$5.6\%$} \\ 
&  &  &  &  &  & 
\end{tabular}%
\end{table}

\subsection{Option Pricing}

Suppose that the price change of a trading asset in day $n$ from the
starting value of zero is $X_{n}$. Thus, after $N$ trading days, the price
of the asset becomes $\xi =X_{1}+X_{2}+\ldots +X_{N}$. For simplicity, we
assume that the asset brings no dividends and consider an European call
option on this asset with strike price $q$ and maturity in $N$ days. The
expected pay-off of the option is $\mathbb{E}\left( \xi -q\right) ^{+}$ and
risk-neutral price is $e^{-rN}\mathbb{E}\left( \xi -q\right) ^{+}$, where $r$
is the risk-free rate. For $q>\frac{\mu ^{2}+\sigma ^{2}}{2\mu }$, \cite%
{LO1987} applies Scarf's inequality to show that
\begin{equation*}
\mathbb{E}\left( \xi -q\right) ^{+}\leq \frac{1}{2}(\mu -q+\sqrt{(q-\mu
)^{2}+\sigma ^{2}}),
\end{equation*}%
which coincides with the aggregation benchmark established in Lemma \ref%
{L:N-NonSharp} of Section \ref{SS:aggregate}.

In practice of options pricing, it is common to use simulations from a
specified parametric distribution. However, this can be very restrictive
becuase an incorrect prior distribution can cause significant losses, as
illustrated by the bundle pricing example. The improved bounds in
Propositions \ref{P:Tail} and \ref{P:ExpectedLoss} are essential in this
setting since traders need a robust estimate of the probability of reaching
the strike price $\xi >q$ 
and a robust estimate of $\mathbb{E}\left\vert \xi -q\right\vert $. 
While \cite{COX1979} introduced the binomial option pricing model,
Proposition \ref{P:ExpectedLoss} complements the seminal work of \citeauthor{COX1979} by deriving 
an upper bound on $\mathbb{E(}\xi -q)^{+}$. Moreover, the introduction of
unequal means and variances, i.e., the scenario with price shocks and price
dynamics, can be achieved by using Equation (\ref{E:NonEqualLoss}) to choose
the parameters for the the binomial option pricing model of \cite{COX1979}.
The equal range property ensures that the lattices are squares with the same
size despite unequal mean and variance. 

As an empirical example, we consider the share price of National Australia
Bank Ltd. (NAB.ASX), which is one of the four largest banks in the country.
We use the data from 10 May to 17 November 2023 as the training data to
compute the mean and standard deviation. This training period happens to
exclude any dividend payments. The average price change on each trading day
is AUD $0.0194$ and standard deviation is AUD $0.2752$. As of 10 May, the
closing price was AUD $26.26$. We assume the strike price is at $q=28.8$ and
we set the discount rate at $r=0$ for convenience. Table \ref%
{Tab:OptionPrices} reports the prices of an European call option computed
for different expiries using three methods, namely, (i) Aggregation, using
the benchmark from Lemma \ref{L:N-NonSharp}, (ii) Improved, using the bound
from Proposition \ref{P:ExpectedLoss}, (iii) Normal Prior, using the
normal distribution with mean $0.0194$ and standard deviation $0.2752$, and (iv) the Black-Scholes solution as presented by \citet[][Tables 1a-1c]{LO1987}.

\begin{table}[tbph]
\caption{Option Prices of NAB Ltd.}
\label{Tab:OptionPrices}\centering
\begin{tabular}{c|ccccc}
Strike Price $q=28.8$ & $N=10$ & $N=30$ & $N=60$ & $N=100$ & $N=200$ \\ 
\hline
Aggregation & $0.078$ & $0.256$ & $0.580$ & $1.108$ & $2.727$ \\ 
Improved & $0.077$ & $0.245$ & $0.532$ & $0.971$ & $2.572$ \\ 
Normal Prior & $0.001$ & $0.069$ & $0.333$ & $0.823$ & $2.313$ \\ 
Black-Scholes & 0.011 & 0.143  & 0.398 & 0.723 & 1.430 
\end{tabular}%
\end{table}

A quick observation reveals that the aggregation-based bound proposed by 
\cite{LO1987} tends to overprice the European call option, while the normal
assumption and Black-Scholes solution result in significant underpricing. 
The improved pricing remains below the aggregation benchmark and above the other two benchmarks for all $N$. From
the computational perspective, the closed form expression and high accuracy
make Proposition \ref{P:ExpectedLoss} an attractive alternative to many
convex algorithms \citep[see][for
updated literature in this area]{JoC2023}. In general, Equation (\ref%
{E:BetaN1}) of Proposition \ref{P:ExpectedLoss} appears more suitable for
European call options as their strike price is often higher than the mean
price while Equation (\ref{E:BetaN2}) of Proposition \ref{P:ExpectedLoss} is
more suitable for European put options as their strike price is often lower
than the mean price.

\section{Conclusion\label{S:Conclusion}}

We develop two sets of results associated with the sum of independent random
variables using only the mean and variance. The results complement earlier
Chebyshev-type results such as \cite{Bentkus2004} and \cite{PIJ2004} 
and provide important new insights, proof strategies and tighter bounds than
those obtained by aggregation. We show significant improvements arising from
using the new bounds in such popular managerial applications as bundle pricing,
inventory management and option pricing.

\ACKNOWLEDGMENT{Please address all correspondence to Artem Prokhorov. Helpful comments from
Rustam Ibragimov and Chung Piaw Teo are gratefully acknowledged.}






\renewcommand{\theequation}{A-\arabic{equation}} 
\setcounter{equation}{0} 


\bibliographystyle{apalike}
\bibliography{BoundIndependentVariable}



\end{document}



\RUNAUTHOR{Li and Prokhorov}

\RUNTITLE{Supplement}

\ARTICLEAUTHORS{\AUTHOR{Zhaolin Li}
\AFF{The University of Sydney Business School, Sydney, NSW2006, Australia, \EMAIL{erick.li@sydney.edu.au}}
\AUTHOR{Artem Prokhorov}
\AFF{The University of Sydney Business School \& CEBA \& CIREQ, Sydney, NSW2006, Australia, \EMAIL{artem.prokhorov@sydney.edu.au}}
} 

\TITLE{Supplement to \\
``Tail Probability and Expected Loss Revisited: \\ Theory and Applications of Semiparametric Bounds''}

\ABSTRACT{This supplement contains technical proofs that help obtain the results in Section 4 of the paper.}



\title{}
\author{}
\maketitle

\bigskip

\renewcommand{\theequation}{A-\arabic{equation}} 
\setcounter{equation}{0} 


\subsection*{Proof of Lemma 3
:}

As a direct result of Corollary 1
, it must hold that $\mu
-\sigma \sqrt{\frac{1-\beta }{\beta }}\leq 0\leq \mu +\sigma \sqrt{\frac{%
\beta }{1-\beta }}$ (otherwise $\beta =\Pr \left( X\leq 0\right) $ cannot be
feasible). Under the extreme distribution in Equation (4.1)
,
the integrand $A_{1}$ equals zero if $X=X^{\prime }$; otherwise, $A_{1}$
equals $H-L=\sigma \sqrt{\frac{\beta }{1-\beta }}+\sigma \sqrt{\frac{1-\beta 
}{\beta }}=\sigma \sqrt{\beta (1-\beta )}$, which is the range of the
two-point distribution. We confirm that the summation equals%
\begin{eqnarray*}
T &=&\sum_{X}\sum_{X^{\prime }}\left[ \left( X-X^{\prime }\right) \left( 
\mathbb{I}_{\{X>0\}}-\mathbb{I}_{\{X^{\prime }>0\}}\right) \right] \Pr
\left( X\right) \Pr \left( X^{\prime }\right) \\
&=&2\left( \sigma \sqrt{\frac{\beta }{1-\beta }}+\sigma \sqrt{\frac{1-\beta 
}{\beta }}\right) \beta \left( 1-\beta \right) =2\sigma \sqrt{\beta -\beta
^{2}},
\end{eqnarray*}%
yielding inequality (4.2).
$\blacksquare $

\subsection*{Proof of Lemma 4
:}

We provide two different proofs of Lemma 4
. The first
proof directly uses the probability mass function of the endogenous Binomial
distribution while the second proof uses Korkine's identity.

\begin{proof}
Using the definition of $\delta _{k}$ and the probability mass function in (4.4)
, we find that%
\begin{equation*}
\mathbb{E}(\xi )^{+}\equiv Z\left( \beta \right) =\sum_{t=k}^{N}\left[ N\mu
+\sigma \left( -\left( N-t\right) \sqrt{\frac{1-\beta }{\beta }}+t\sqrt{%
\frac{\beta }{1-\beta }}\right) \right] \frac{N!\left( 1-\beta \right)
^{t}\beta ^{N-t}}{t!(N-t)!}.
\end{equation*}%
We define a sequence of $\left\{ y_{t}\right\} $ as follows. For $t=N$, it
holds that%
\begin{equation*}
y_{N}=N\sqrt{\frac{\beta }{1-\beta }}\left( 1-\beta \right) ^{N}=N\left(
1-\beta \right) ^{N-1}\sqrt{\beta \left( 1-\beta \right) }.
\end{equation*}%
For $t\leq N-1$, it holds that%
\begin{eqnarray*}
y_{t} &=&\left( -\left( N-t\right) \sqrt{\frac{1-\beta }{\beta }}+t\sqrt{%
\frac{\beta }{1-\beta }}\right) \frac{N!\left( 1-\beta \right) ^{t}\beta
^{N-t}}{t!(N-t)!} \\
&=&-\left( N-t\right) \sqrt{\frac{1-\beta }{\beta }}\frac{N!\left( 1-\beta
\right) ^{t}\beta ^{N-t}}{t!(N-t)!}+t\sqrt{\frac{\beta }{1-\beta }}\frac{%
N!\left( 1-\beta \right) ^{t}\beta ^{N-t}}{t!(N-t)!} \\
&=&-N\sqrt{\beta \left( 1-\beta \right) }\frac{\left( N-1\right) !\left(
1-\beta \right) ^{t}\beta ^{N-t-1}}{t!(N-t-1)!}+N\sqrt{\beta \left( 1-\beta
\right) }\frac{\left( N-1\right) !\left( 1-\beta \right) ^{t-1}\beta ^{N-t}}{%
\left( t-1\right) !(N-t)!}.
\end{eqnarray*}%
Contrasting $y_{t}$ and $y_{t+1}$ (for $t\leq N-1$), we find that the
negative term of $y_{t}$ equals the positive term of $y_{t+1}$. Hence, the
summation equals%
\begin{eqnarray*}
&&y_{k}+y_{k+1}+...+y_{N} \\
&=&-N\sqrt{\beta \left( 1-\beta \right) }\frac{\left( N-1\right) !\left(
1-\beta \right) ^{k}\beta ^{N-k-1}}{k!(N-k-1)!}+N\sqrt{\beta \left( 1-\beta
\right) }\frac{\left( N-1\right) !\left( 1-\beta \right) ^{k-1}\beta ^{N-k}}{%
\left( k-1\right) !(N-k)!} \\
&&-N\sqrt{\beta \left( 1-\beta \right) }\frac{\left( N-1\right) !\left(
1-\beta \right) ^{k+1}\beta ^{N-k-2}}{(k+1)!(N-k-2)!}+N\sqrt{\beta \left(
1-\beta \right) }\frac{\left( N-1\right) !\left( 1-\beta \right) ^{k}\beta
^{N-k-1}}{k!(N-k-1)!} \\
&&...-N\sqrt{\beta \left( 1-\beta \right) }\frac{\left( N-1\right) !\left(
1-\beta \right) ^{N-1}}{(N-1)!}+N\sqrt{\beta \left( 1-\beta \right) }\frac{%
\left( N-1\right) !\left( 1-\beta \right) ^{N-2}\beta }{(N-2)!1!} \\
&&+N\left( 1-\beta \right) ^{N-1}\sqrt{\beta \left( 1-\beta \right) } \\
&=&N\sqrt{\beta \left( 1-\beta \right) }\frac{\left( N-1\right) !\left(
1-\beta \right) ^{k-1}\beta ^{N-k}}{\left( k-1\right) !(N-k)!},
\end{eqnarray*}%
whereby only the positive term of $y_{t}$ is not cancelled out. Thus, the
expected loss equals%
\begin{eqnarray*}
Z\left( \beta \right) &=&\sigma \left( y_{k}+y_{k+1}+...+y_{N}\right) +N\mu
\sum_{t=k}^{N}\frac{N!\left( 1-\beta \right) ^{t}\beta ^{N-t}}{t!(N-t)!} \\
&=&N\sigma \sqrt{\beta \left( 1-\beta \right) }\frac{\left( N-1\right)
!\left( 1-\beta \right) ^{k-1}\beta ^{N-k}}{\left( k-1\right) !(N-k)!}+N\mu
\sum_{t=k}^{N}\frac{N!\left( 1-\beta \right) ^{t}\beta ^{N-t}}{t!(N-t)!},
\end{eqnarray*}%
which proves Equation (4.6)
.
\end{proof}

\begin{proof}
We apply Equation (4.13) 
to compute the expected loss under the
endogenous Binomial distribution. We observe that if $\beta \in \left[
\delta _{k},\delta _{k-1}\right] $, $\left( N-t\right) L+tH$ is positive for 
$t\geq k$ but is negative for $t\leq k-1$. Thus, when $\xi _{(i)}=\left(
N-k\right) L+\left( k-1\right) H$, it holds that $H+\xi _{(i)}=\left(
N-k\right) L+kH>0$ and $L+\xi _{(i)}=\left( N-k+1\right) L+\left( k-1\right)
H\leq 0$. The coefficient equals $\left( \mathbb{I}_{\{H+\xi _{(i)}>0\}}-%
\mathbb{I}_{\{L+\xi _{(i)}>0\}}\right) =1$. This event occurs with
probability $\frac{(N-1)!\left( 1-\beta \right) ^{k-1}\beta ^{N-k}}{%
(k-1)!(N-k)!}$. For all the other $\xi _{(i)}\neq \left( N-k\right) L+\left(
k-1\right) H$, the two indicators have the same value, making the
coefficient $\left( \mathbb{I}_{\{H+\xi _{(i)}>0\}}-\mathbb{I}_{\{L+\xi
_{(i)}>0\}}\right) =0$. Hence, we find that when $\beta \in \left[ \delta
_{k},\delta _{k-1}\right] $, Equation (4.13) can be written as follows:
\begin{eqnarray*}
Z\left( \beta \right) &=&N\mu \sum_{t=k}^{N}\frac{N!\beta ^{N-t}\left(
1-\beta \right) ^{t}}{t!\left( N-t\right) !} \\
&&+\frac{N}{2}2\sigma \left( \sqrt{\frac{\beta }{1-\beta }}+\sqrt{\frac{%
1-\beta }{\beta }}\right) \beta \left( 1-\beta \right) \frac{\left(
N-1\right) !\left( 1-\beta \right) ^{k-1}\beta ^{N-k}}{\left( k-1\right)
!\left( N-k\right) !},
\end{eqnarray*}%
which is identical to Equation (4.6)
.
\end{proof}

\subsection*{Proof of Lemma 5
:}

We recall that over the interval $\left[ \delta _{k},\delta _{k-1}\right] $,
the piecewise objective function equals 
\begin{equation*}
T\left( k,\beta \right) =N\sigma \sqrt{\beta \left( 1-\beta \right) }\frac{%
\left( N-1\right) !\left( 1-\beta \right) ^{k-1}\beta ^{N-k}}{\left(
k-1\right) !\left( N-k\right) !}.
\end{equation*}%
It is more convenient to take logarithm and consider $\ln (T(k,\beta ))=\ln
(N\sigma )+\left( k-\frac{1}{2}\right) \ln \left( 1-\beta \right) +\left(
N-k+\frac{1}{2}\right) \ln \beta $. The first order condition yields that 
\begin{equation*}
\frac{\partial \ln T\left( k,\beta \right) }{\partial \beta }=-\frac{\left(
2N-2k-2N\beta +1\right) }{2\beta \left( \beta -1\right) }=0,
\end{equation*}%
indicating that $\beta _{k}^{\ast }=\frac{1}{2N}\left( 2N-2k+1\right) $.
Substituting $\beta _{k}^{\ast }$ into $T\left( k,\beta \right) $, we obtain
the local optimal objective value $T^{\ast }\left( k\right) $ shown in
Equation (4.7)
. The second order condition yields that 
\begin{equation*}
\frac{\partial ^{2}\ln T\left( k,\beta \right) }{\partial \beta ^{2}}=-\frac{%
1}{2\beta ^{2}\left( \beta -1\right) ^{2}}\left( 2N-2k-2\beta -4N\beta
+4k\beta +2N\beta ^{2}+1\right) .
\end{equation*}%
The numerator is a convex and quadratic function with respect to $\beta $.
The determinant of this quadratic equation equals 
\begin{equation*}
\Delta =\left( -2-4N+4k\right) ^{2}-4\cdot 2N\cdot \left( 2N-2k+1\right)
=-4\left( 2k-1\right) \left( 2N-2k+1\right) \allowbreak <0
\end{equation*}%
due to $1\leq k\leq N$. Thus, the numerator is always positive, meaning that 
$\frac{\partial ^{2}\ln T\left( k,\beta \right) }{\partial \beta ^{2}}<0$
and making $\ln T\left( k,\beta \right) $ a log-concave function with
respect to $\beta $. We conclude that $\beta _{k}^{\ast }=\frac{1}{2N}\left(
2N-2k+1\right) $ is a local optimal solution over the interval $\left[
\delta _{k},\delta _{k-1}\right] $.

(i) The symmetry property as illustrated in Figure 1(a) of the paper 
 is
trivial due to the relationship 
\begin{equation*}
\beta _{k}^{\ast }=\frac{1}{2N}\left( 2N-2k+1\right) =1-\frac{1}{2N}\left(
2N-2\left( N-k\right) +1\right) =\allowbreak 1-\beta _{N-k}^{\ast }
\end{equation*}%
and the symmetry of $T\left( k,\beta \right) =T\left( N-k,1-\beta \right) $.

(ii) To prove the log-convex property, we take logarithm such that 
\begin{eqnarray*}
\ln T_{k} &=&\ln \Gamma \left( N\right) -\ln \Gamma \left( k\right) -\ln
\Gamma \left( N-k+1\right) \\
&&+\left( k-\frac{1}{2}\right) \ln \left( \frac{2k-1}{2N}\right) +\left( N-k+%
\frac{1}{2}\right) \ln \left( 1-\frac{2k-1}{2N}\right) ,
\end{eqnarray*}%
where $\Gamma \left( k\right) =\left( k-1\right) !$ is the Gamma function.
Since 
\begin{equation*}
\frac{\partial ^{2}}{\partial k^{2}}\left( \left( k-\frac{1}{2}\right) \ln
\left( \frac{2k-1}{2N}\right) +\left( N-k+\frac{1}{2}\right) \ln \left( 1-%
\frac{2k-1}{2N}\right) \right) =\frac{1}{k-\frac{1}{2}}+\frac{1}{N-k+\frac{1%
}{2}},
\end{equation*}%
and $\ln \Gamma \left( N\right) $ is a constant, our task is to show that 
\begin{equation*}
G=\frac{\partial ^{2}\ln \Gamma \left( k\right) }{\partial k^{2}}+\frac{%
\partial ^{2}\ln \Gamma \left( N-k+1\right) }{\partial k^{2}}<\frac{1}{k-%
\frac{1}{2}}+\frac{1}{N-k+\frac{1}{2}}.
\end{equation*}

It is known that $\frac{\partial ^{2}\ln \Gamma \left( k\right) }{\partial
k^{2}}=\psi ^{(1)}\left( k\right) $ is the Trigamma function satisfying $%
\psi ^{(1)}\left( k\right) \approx \frac{1}{k}+\frac{1}{2k^{2}}+\frac{1}{%
6k^{3}}-\frac{1}{30x^{5}}+...$. As the fourth term has a negative
coefficient, we use the first three terms of the expansion to construct an
upper bound on $G$ as follows: 
\begin{equation*}
G\leq \frac{1}{k}+\frac{1}{2k^{2}}+\frac{1}{6k^{3}}+\frac{1}{N-k+1}+\frac{1}{%
\left( N-k+1\right) ^{2}}+\frac{1}{6\left( N-k+1\right) ^{3}},
\end{equation*}%
which is consistent with Theorem 4 in \cite{Gordon1994}. Due to symmetry, it
suffices to show that%
\begin{equation*}
\frac{1}{k}+\frac{1}{2k^{2}}+\frac{1}{6k^{3}}-\frac{1}{k-\frac{1}{2}}=-\frac{%
k+1}{6k^{3}\left( 2k-1\right) }<0,
\end{equation*}%
for any $1\leq k\leq N$. We conclude that $G<\frac{1}{k-\frac{1}{2}}+\frac{1%
}{N-k+\frac{1}{2}}$, which ensures that $\frac{\partial ^{2}}{\partial k^{2}}%
\left( \ln (T_{k})\right) >0$. $\blacksquare $

\subsection*{Proof of Proposition 3
:}

When optimizing $Z_{1}(\beta )$, we solve the following first order
condition: 
\begin{eqnarray*}
\frac{\partial Z_{1}\left( \beta \right) }{\partial \beta } &=&\frac{%
\partial }{\partial \beta }\left( N\mu \left( 1-\beta ^{N}\right) +N\sigma 
\sqrt{\beta \left( 1-\beta \right) }\beta ^{N-1}\right) \\
&=&-\frac{N\beta ^{N-1}}{2\sqrt{\beta \left( 1-\beta \right) }}\left( \sigma
-2N\sigma +2N\sigma \beta +2N\mu \sqrt{\beta \left( 1-\beta \right) }\right)
=0,
\end{eqnarray*}%
which results in Equation (4.9)
. The candidate solution $\hat{%
\beta}_{1}$ is in the interior of $\left[ \delta _{1},1\right] $, meaning
that 
\begin{equation*}
-\left( N-1\right) \sigma \sqrt{\frac{1-\hat{\beta}_{1}}{\hat{\beta}_{1}}}%
+\sigma \sqrt{\frac{\hat{\beta}_{1}}{1-\hat{\beta}_{1}}}>0.
\end{equation*}%
Thus, under $\hat{\beta}_{1}$, only when all $X_{n}=\mu -\sigma \sqrt{\frac{%
1-\hat{\beta}_{1}}{\hat{\beta}_{1}}}$, can the sum $\xi $ be negative. The
expected loss indeed equals $Z^{\ast }=N\mu -N\hat{\beta}_{1}^{N}\left( \mu
-\sigma \sqrt{\frac{1-\hat{\beta}_{1}}{\hat{\beta}_{1}}}\right) $ under this
extreme two-point distribution. Similarly, by optimizing $Z_{N}\left( \beta
\right) $ we solve the following first order condition:%
\begin{eqnarray*}
\frac{\partial Z_{N}\left( \beta \right) }{\partial \beta } &=&\frac{%
\partial }{\partial \beta }\left( N\mu \left( 1-\beta \right) ^{N}+N\sigma 
\sqrt{\beta \left( 1-\beta \right) }\left( 1-\beta \right) ^{N-1}\right) \\
&=&-\frac{N\left( 1-\beta \right) ^{N}\sqrt{\beta \left( 1-\beta \right) }}{%
2\beta \left( 1-\beta \right) ^{2}}\left( 2N\sigma \beta -\sigma +2N\mu 
\sqrt{\beta \left( 1-\beta \right) }\right) =0,
\end{eqnarray*}%
which results in Equation (4.10).

Next, we investigate the sign of $g\left( \mu \right) =\frac{1}{N}%
\left(Z_{1}(\hat{\beta}_{1})-Z_{N}(\hat{\beta}_{N})\right)$. When $\mu =0$,
Lemma 5 
has shown that $g\left( 0\right) =0$. Applying
the envelope theorem, we find that%
\begin{equation*}
g^{\prime }\left( \mu \right) =\frac{1}{N}\frac{\partial Z_{1}(\beta )}{%
\partial \mu }|_{\beta =\hat{\beta}_{1}}-\frac{1}{N}\frac{\partial
Z_{N}(\beta )}{\partial \mu }|_{\beta =\hat{\beta}_{N}}=1-\hat{\beta}%
_{1}^{N}-\left( 1-\hat{\beta}_{N}\right) ^{N}.
\end{equation*}%
At the point $\mu =0$, it holds that%
\begin{equation*}
g^{\prime }\left( \mu \right) |_{\mu =0}=1-\left( \frac{2N-1}{2N}\right)
^{N}-\left( 1-\frac{1}{2N}\right) ^{N}<0.
\end{equation*}%
Thus, when $\Delta \mu >0$ (which occurs when $\mu $ increases from zero to
a positive number), $g\left( \mu \right) \simeq \left[ g^{\prime }\left( \mu
\right) |_{\mu =0}\right] \Delta \mu <0$, suggesting that $Z_{1}(\hat{\beta}%
_{1})-Z_{N}(\hat{\beta}_{N})<0$, making $\hat{\beta}_{N}$ a better solution
than $\hat{\beta}_{1} $. Likewise, when $\Delta \mu <0$ (which occurs when $%
\mu $ decreases from zero to a negative number), $g\left( \mu \right) \simeq %
\left[ g^{\prime }\left( \mu \right) |_{\mu =0}\right] \Delta \mu >0$,
suggesting that $Z_{1}(\hat{\beta}_{1})-Z_{N}(\hat{\beta}_{N})>0$, making $%
\hat{\beta}_{1}$ a better solution than $\hat{\beta}_{N}$.

Using the same method, we can contrast $\max \left\{ Z_{1}(\hat{\beta}%
_{1}),Z_{N}(\hat{\beta}_{N})\right\} $ with any other local optimum
objective value $Z_{k}(\hat{\beta}_{k})$, where $k\in \left\{
2,3,...,N-1\right\} $. We notice that%
\begin{eqnarray*}
Z_{k}\left( \beta \right) &=&\frac{N!}{(N-k)!(k-1)!}\sigma \sqrt{\beta
(1-\beta )}(1-\beta )^{k-1}\beta ^{N-k} \\
&&+N\mu \left( 1-\beta ^{N}-N\beta ^{N-1}(1-\beta )-...-\frac{N!}{%
(N-k+1)!(k-1)!}\beta ^{N-k+1}(1-\beta )^{k-1}\right) .
\end{eqnarray*}%
The first order condition yields that%
\begin{equation*}
\frac{\partial Z_{k}\left( \beta \right) }{\partial \beta }=-\frac{\sigma }{2%
}\frac{N!\beta ^{N-k+1}\left( 1-\beta \right) ^{k-1}}{(N-k)!(k-1)!\sqrt{%
\beta \left( 1-\beta \right) }}\left[ \left( 2k-1\right) \sigma -2N\sigma
+2N\sigma \beta +2N\mu \sqrt{\beta \left( 1-\beta \right) }\right] =0.
\end{equation*}%
Thus, we find that%
\begin{equation*}
\hat{\beta}_{k}=\frac{\sigma ^{2}\left( 2N-2k+1\right) +N\mu ^{2}-\mu \sqrt{%
N^{2}\mu ^{2}+\sigma ^{2}\left( 2k-1\right) \left( 2N-2k+1\right) }}{%
2N\left( \sigma ^{2}+\mu ^{2}\right) }.
\end{equation*}

In a special case with $\mu =0$, we recover the same result in the proof of
Lemma 5 
that $\hat{\beta}_{k}=\frac{2N-2k+1}{2N}$. When $%
\mu <0$ (which implies $\Delta \mu <0$), we find that%
\begin{eqnarray*}
\frac{1}{N}\frac{\partial }{\partial \mu }Z_{k}(\hat{\beta}_{k})|_{\mu =0}
&=&1-\hat{\beta}_{k}^{N}-N\hat{\beta}_{k}^{N-1}(1-\hat{\beta}_{k})-...-\frac{%
N!}{(N-k+1)!(k-1)!}\hat{\beta}_{k}^{N-k+1}(1-\hat{\beta}_{k})^{k-1} \\
&=&1-\left( \frac{2N-2k+1}{2N}\right) ^{N}-...-\frac{N!}{(N-k+1)!(k-1)!}%
\left( \frac{2N-2k+1}{2N}\right) ^{N-k+1}\left( \frac{2k-1}{2N}\right) ^{k-1}
\\
&>&1-\left( \frac{2N-1}{2N}\right) ^{N}=\frac{\partial }{\partial \mu }Z_{1}(%
\hat{\beta}_{1})|_{\mu =0}.
\end{eqnarray*}%
At the point $\mu =0$, Lemma 5 
already shows that $Z_{1}(%
\hat{\beta}_{1})>Z_{k}(\hat{\beta}_{k})$. We find that $Z_{1}(\hat{\beta}%
_{1})-Z_{k}(\hat{\beta}_{k})\simeq \frac{\partial }{\partial \mu }\left(
Z_{1}(\hat{\beta}_{1})-Z_{k}(\hat{\beta}_{k})\right) |_{\mu =0}\Delta \mu >0$%
, making $\hat{\beta}_{1}$ the global optimal solution for $Z\left( \beta
\right) $ when $\mu <0$. Similarly, we find that when $\mu >0$, $Z_{N}(\hat{%
\beta}_{N})>Z_{k}(\hat{\beta}_{k})$.

In summary, the global optimal solution is 1) $\hat{\beta}_{1}$ when $\mu <0$
or 2) $\hat{\beta}_{N}$ when $\mu >0$. Certainly, Lemma 5 
already shows that when $\mu =0$, there exist two global optimal solutions $%
\hat{\beta}_{1}=\frac{2N-1}{2N}$ and $\hat{\beta}_{N}=\frac{1}{2N}$. $%
\blacksquare $

\section*{Proof of Theorem 1
}

As the sum $\xi $ satisfies $\mathbb{E}\left( \xi \right) =N\mu $ and $%
Var(\xi )=N^{2}\mu ^{2}+N\sigma ^{2}$, we denote the cumulative distribution
function of $\xi $ by $G\left( \cdot \right) $ and define the following
ambiguity set:%
\[
\Omega _{0}=\left\{ G|\int_{-\infty }^{\infty }\xi ^{t}dG\left( \xi \right)
=m_{t}^{\prime},t=0,1,2\right\} ,
\]%
where $m_{0}^{\prime}=1$, $m_{1}^{\prime}=N\mu $, and $m_{2}^{\prime}=N^{2}\mu ^{2}+N\sigma ^{2}$. With
independence, the ambiguity set is the following:%
\[
\Omega _{1}=\left\{ G|G\left( \cdot \right) =F_{N}\left( \cdot \right)
,\int_{-\infty }^{\infty }X^{t}dF\left( \xi \right) =m_{t},t=0,1,2\right\} ,
\]%
where $m_{0}=1$, $m_{1}=\mu $, $m_{2}=\mu
^{2}+\sigma ^{2}$, and $F_{N}$ is the $N$-fold convolution of $F$. The
similarity between sets $\Omega _{0}$ and $\Omega _{1}$ is that they include
distribution functions $G\left( \cdot \right) $ that satisfy the
mean-variance conditions on the sum $\xi $. The difference is that any
distribution from set $\Omega _{1}$ can be written as the $N$-fold
convolution of some iid distribution $F$; while some distributions from set $%
\Omega _{0}$ cannot. For
instance, if $\xi $ follows the following two-point distribution:%
\begin{equation}
\left\{
\begin{array}{l}
\Pr \left( \tilde{\xi}=N\mu -\sqrt{N}\sigma \sqrt{\frac{1-\gamma }{\gamma }}%
\right) =\gamma , \\
\Pr \left( \tilde{\xi}=N\mu +\sqrt{N}\sigma \sqrt{\frac{\gamma }{1-\gamma }}%
\right) =1-\gamma ,%
\end{array}%
\right.   \label{E:2PointXi}
\end{equation}%
we cannot find any iid distribution $F$ such that $G\left( \xi \right)
=F_{N}\left( \xi \right) $. The reason is simple and intuitive. Suppose that
each $\theta _{n}$ has at least two realized values. The \textquotedblleft
presumed" sum $\xi =\theta _{1}+...+\theta _{N}$ must have at least $N+1\geq
3$ realized values for any $N\geq 2$. In contrast, the two-point
distribution in equation (\ref{E:2PointXi}) has only two realized values. We
conclude that $\Omega _{0}\supset \Omega _{1}$.

We seek to choose a distribution from the smaller set $\Omega _{1}$ to
maximize the expected linear loss: $Z_{I}=\sup_{G\in \Omega _{1}}\mathbb{E}%
\left( \xi -q\right) ^{+}$, where the subscript $I$ indicates independence.
We now modify the larger set $\Omega _{0}$. Let $G$ be a candidate distribution from set $\Omega _{0}$ and $%
G^{-1}\left( \gamma \right) $ be the inverse. We require that%
\begin{equation}
L_{N}\equiv N\mu -N\sigma \sqrt{\frac{1-\gamma ^{\frac{1}{N}}}{\gamma ^{%
\frac{1}{N}}}}\leq G^{-1}\left( \gamma \right) \leq N\mu +N\sigma \sqrt{%
\frac{1-(1-\gamma )^{\frac{1}{N}}}{(1-\gamma )^{\frac{1}{N}}}}\equiv H_{N}.
\label{E:PercentileConstraints}
\end{equation}%
We observe that constraints (\ref{E:PercentileConstraints}) are the
necessary conditions for independence. Specifically, any distribution from
set $\Omega _{1}$ must satisfy constraints (\ref{E:PercentileConstraints})
as Corollary 1 
suggested. But some distributions from set $%
\Omega _{0}$, despite satisfying constraints (\ref{E:PercentileConstraints}%
), may not be an element of the set $\Omega _{1}$.

We consider another ambiguity set%
\[
\Omega _{0}^{\prime }=\left\{ G|\int_{-\infty }^{\infty }\xi ^{t}dG\left(
\xi \right) =m_{t}^{\prime},t=0,1,\text{ and }\int_{-\infty }^{\infty }\xi
^{2}dG\left( \xi \right) \leq N^{2}\mu ^{2}+N\sigma ^{2}\right\} .
\]%
It must hold that $\Omega _{0}^{\prime }\supset \Omega _{0}$. The
aggregation model yields that $\bar{Z}_{0}^{\ast
}=\sup_{G\in \Omega _{0}^{\prime }}\mathbb{E}\left( \xi -q\right)
^{+}=Z_{0}^{\ast }=\sup_{G\in \Omega _{0}}\mathbb{E}\left( \xi -q\right) ^{+}
$. Specifically, part (b) of Lemma 2 
implies that the bound on $\mathbb{E}(\xi-q)^+$ is increasing in $\sigma$ and hence,  letting the variance to be $\leq \sigma^2$ rather than $=\sigma^2$ will not alter the result. Next, we consider the following relaxed model:%
\[
\begin{tabular}{ll}
$Z_{r}=\sup_{G\in \Omega _{0}^{\prime }}\mathbb{E}\left( \xi -q\right) ^{+}$
& $\text{s.t. }L_{N}\leq G^{-1}\left( \gamma \right) \leq H_{N},\forall
\gamma \in (0,1).$%
\end{tabular}%
\]%
Conceptually, we relax the variance constraint from $=$ to $\leq $ but at
the same time, we add the percentile constraints. It is trivial to verify
that $Z_{r}^{\ast }\geq Z_{I}^{\ast }$ as $\Omega _{0}^{\prime }\cap \left\{
G|L_{N}\leq G^{-1}\left( \gamma \right) \leq H_{N}\right\} \supset \Omega
_{1}$ (i.e., any feasible distribution from the smaller set $\Omega _{1}$
must be an element of $\Omega _{0}^{\prime }$). The relaxed model $Z_{r}$ is
a one-dimensional model regarding the sum $\xi $. We observe that $\left(
\xi -q\right) ^{+}$ is continuous and convex in $\xi $. Thus, the extreme
distribution for $\xi $ is a two-point distribution satisfying%
\begin{equation}\label{E:TwoPointDominatingXi}
\left\{
\begin{array}{l}
\Pr \left( \xi =N\mu -\sqrt{N}\sigma _{0}\sqrt{\frac{1-\gamma}{\gamma }} \equiv L_{\gamma} \right) =\gamma , \\
\Pr \left( \xi =N\mu +\sqrt{N}\sigma _{0}\sqrt{\frac{\gamma }{1-\gamma }} \equiv H_{\gamma} \right) =1-\gamma ,%
\end{array}%
\right.
\end{equation}%
where $\sigma _{0}\leq \sigma $. We recall that the percentile must be a root for equation $G(x)=\gamma$. The cumulative distribution function for equation (\ref{E:TwoPointDominatingXi}) is a step function satisfying $G(x)=\gamma$ when $L_{\gamma} \leq x < H_{\gamma}$. Thus, the percentile constraints will fail to hold only when both  $L_{\gamma}<L_N$ and $H_N<H_{\gamma}$ occur. If either $L_N \leq L_{\gamma}$ or $H_{\gamma} \leq H_N$ occurs, we can identify an iid distribution such that the sum satisfies the percentile constraints. Hence, the following relaxed model $Z_{r}$ is our focus:
\begin{eqnarray*}
Z_{r} &=&\sup_{G\in \Omega _{0}}E\left( \xi -s\right) ^{+}=\max_{\substack{ %
\gamma \in (0,1) \\ \sigma _{0}\leq \sigma }}\left\{ \left( 1-\gamma \right)
\left( N\mu +\sqrt{N}\sigma _{0}\sqrt{\frac{\gamma }{1-\gamma }}-s\right)
\right\}  \\
\text{s.t. } &&L_{N}\leq N\mu -N\sigma _{0}\sqrt{\frac{1-\gamma }{\gamma }}%
\text{ or }N\mu +\sqrt{N}\sigma _{0}\sqrt{\frac{\gamma }{1-\gamma }}\leq
H_{N}\text{.}
\end{eqnarray*}%
Evidently, $\sigma _{0}\leq \sigma $ is non-binding; otherwise, the sum of
dependent variables attains the same bound as the sum of independent
variables (and the bound reduces to the aggregate bound in Lemma 2
).

We obtain the following results for $Z_{r}$. 1) When $L_{N}\leq N\mu
-N\sigma _{0}\sqrt{\frac{1-\gamma }{\gamma }}$ is binding, it holds that $%
\sigma _{0}=\frac{N\mu -L_{N}}{\sqrt{N}\sqrt{\frac{1-\gamma }{\gamma }}}$
and hence,%
\begin{eqnarray}
Z_{r} &=&\left( 1-\gamma \right) \left( N\mu +\sqrt{N}\sigma _{0}\sqrt{\frac{%
\gamma }{1-\gamma }}-s\right) =\left( 1-\gamma \right) \left( N\mu +\left(
N\mu -L_{N}\right) \frac{\gamma }{1-\gamma }-s\right)   \nonumber \\
&=&\left( N\mu -s\right) +\gamma \left( s-N\mu +N\sigma \sqrt{\frac{1-\gamma
^{\frac{1}{N}}}{\gamma ^{\frac{1}{N}}}}\right) .  \label{E:Z1-Leftmost}
\end{eqnarray}%
Similarly, when $N\mu +\sqrt{N}\sigma _{0}\sqrt{\frac{\gamma }{1-\gamma }}%
\leq H_{N}$ is binding, it holds that%
\begin{eqnarray}
Z_{r} &=&\left( 1-\gamma \right) \left( N\mu +\sqrt{N}\frac{H_{N}-N\mu }{%
\sqrt{N}\sqrt{\frac{\gamma }{1-\gamma }}}\sqrt{\frac{\gamma }{1-\gamma }}%
-s\right) =\left( 1-\gamma \right) \left( H_{N}-s\right)   \nonumber \\
&=&\left( 1-\gamma \right) \left( N\mu +N\sigma \sqrt{\frac{1-(1-\gamma )^{%
\frac{1}{N}}}{(1-\gamma )^{\frac{1}{N}}}}-s\right) .  \label{E:ZN-RightMost}
\end{eqnarray}
In equation (\ref{E:Z1-Leftmost}), we let $\beta =\gamma ^{\frac{1}{N}}$ to
recover the objective function $Z_{1}\left( \beta \right) $ that we develop
in the proof of Proposition 3; 
in equation (\ref{E:Z1-Leftmost}), we
let $1-\beta =\left( 1-\gamma \right) ^{\frac{1}{N}}$ to recover the
objective function $Z_{N}\left( \beta \right) $ that we develop in the proof of
Proposition 3.
We find that $Z_{r}=\max \left\{
\max_{\beta }Z_{1}\left( \beta \right) ,\max_{\beta }Z_{N}\left( \beta
\right) \right\} $. Either $Z_{1}\left( \beta \right) $ or
$Z_{N}\left( \beta \right) $ can be attained by a distribution from set $\Omega_1$, meaning that the equal sign in $Z_{r}\geq Z_{I}$ must hold. In summary, the
two-point distributions developed in Proposition 3 
are
indeed the extreme distributions attaining the maximal expected linear loss.
$\blacksquare $ 

\bibliographystyle{apalike}
\bibliography{BoundIndependentVariable}